\documentclass[11pt,a4paper]{article}
\usepackage{graphicx} 
\usepackage{xcolor}
\usepackage{fullpage}
\usepackage{algorithm}
\usepackage{algorithmic}
\usepackage{amssymb} 
\usepackage{amsthm}
\usepackage{amsmath}

\usepackage{authblk}

\usepackage[margin=1in]{geometry}

\newtheorem{theorem}{Theorem}
\newtheorem{lemma}{Lemma}

\title{Streaming Maximal Matching with Bounded Deletions}
\author[1]{Sanjeev Khanna}
\author[2]{Christian Konrad\footnote{C.K. was supported by EPSRC New Investigator Award EP/V010611/1.}}
\author[3]{Jacques Dark\footnote{Part of the work of J.D. was done while at the University of Warwick,
supported by an EMEA Microsoft Research PhD studentship and European Research Council grant
ERC-2014-CoG 647557.}}
\affil[1]{Department of Computer and Information Science, University of Pennsylvania, Philadelphia, PA, US,
\texttt{sanjeev@cis.upenn.edu}
}
\affil[2]{Department of Computer Science,
University of Bristol, Bristol, UK,
\texttt{christian.konrad@bristol.ac.uk}}

\affil[3]{Unaffiliated Researcher, London, UK,
\texttt{thejdark@gmail.com}}

\date{}

\DeclareMathOperator{\poly}{poly}
\DeclareMathOperator{\Exp}{\mathbb{E}}

\newcommand{\D}{\mathsf{\bf D}}

\newcommand{\LHMM}{{\sf LHMM}}

\allowdisplaybreaks

\begin{document}
\maketitle
\thispagestyle{empty} 

\vspace{-0.5cm}
\begin{abstract}
We initiate the study of the \textsf{Maximal Matching} problem in bounded-deletion graph streams. In this setting, a graph $G$ is revealed as an arbitrary sequence of edge insertions and deletions, where the number of insertions is unrestricted but the number of deletions is guaranteed to be at most $K$, for some given parameter $K$. The single-pass streaming space complexity of this problem is known to be $\Theta(n^2)$ when $K$ is unrestricted, where $n$ is the number of vertices of the input graph. 
In this work, we present new randomized and deterministic algorithms and matching lower bound results that together give a tight understanding (up to poly-log factors) of how the space complexity of \textsf{Maximal Matching} evolves as a function of the parameter $K$: The randomized space complexity of this problem is $\tilde{\Theta}(n \cdot \sqrt{K})$, while the deterministic space complexity is $\tilde{\Theta}(n \cdot K)$.
We further show that if we relax the maximal matching requirement to an $\alpha$-approximation to \textsf{Maximum Matching}, for any constant $\alpha > 2$, then the space complexity for both, deterministic and randomized algorithms, strikingly changes to $\tilde{\Theta}(n + K)$. 

\smallskip

A key conceptual contribution of our work that underlies all our algorithmic results is the introduction of the {\em hierarchical maximal matching} data structure, which computes a hierarchy of $L$ maximal matchings on the substream of edge insertions, for an integer $L$. This deterministic data structure allows recovering a \textsf{Maximal Matching} even in the presence of up to $L-1$ edge deletions, which immediately yields an optimal deterministic algorithm with space $\tilde{O}(n \cdot K)$. To reduce the space to $\tilde{O}(n \cdot \sqrt{K})$, we compute only $\sqrt{K}$ levels of our hierachical matching data structure and utilize a randomized linear sketch, i.e., our {\em matching repair data structure}, to repair any damage due to edge deletions. Using our repair data structure, we show that the level that is least affected by deletions can be repaired back to be globally maximal. The repair data structure is computed independently of the hierarchical maximal matching data structure and stores information for vertices at different scales with a gradually smaller set of vertices storing more and more information about their incident edges. The repair process then makes progress either by rematching a vertex to a previously unmatched vertex, or by strategically matching it to another matched vertex whose current mate is in a better position to find a new mate in that we have stored more information about its incident edges. 

\smallskip

Our lower bound result for randomized algorithms is obtained by establishing a lower bound for a generalization of the well-known \textsf{Augmented-Index} problem in the one-way two-party communication setting that we refer to as \textsf{Embedded-Augmented-Index}, and then showing that an instance of \textsf{Embedded-Augmented-Index} reduces to computing a maximal matching in bounded-deletion streams. To obtain our lower bound for deterministic algorithms, we utilize a compression argument to show that a deterministic algorithm with space $o(n \cdot K)$ would yield a scheme to compress a suitable class of graphs below the information-theoretic threshold.
\end{abstract}

\clearpage
\pagenumbering{arabic} 

\newpage

\section{Introduction}
In the streaming model of computation, an algorithm is tasked with computing the solution to a given problem by performing a single pass over the input data while maintaining a memory of sublinear size in the input. 

\smallskip
\noindent \textbf{Graph Streams.} Streaming algorithms have been studied for more than 25 years \cite{ams96}, and, since early on, significant attention has been devoted to the study of graph problems in this setting \cite{hrr98,fkmsz04}. The dominant model is the {\em insertion-only} model, where the input stream consists of a sequence of edges that make up a graph. In 2012, Ahn et al. \cite{agm12} initiated the study of graph problems in the {\em insertion-deletion} setting, where the input stream consists of a sequence of edge insertions and deletions so that previously inserted edges can be deleted again. Subsequent research has revealed that some problems are equally difficult to solve in both the insertion-only and the insertion-deletion models in that algorithms require roughly the same amount of space, e.g., \textsf{Connectivity} \cite{agm12} and \textsf{Maximum Independent Set} \cite{hssw12}, while others are significantly harder to solve in the context of deletions, e.g., \textsf{Maximal Matching},   \textsf{Maximum Matching}, and \textsf{Minimum Vertex Cover} \cite{akl16,dk20}.

Handling deletions requires a very different algorithmic toolkit than when deletions are not allowed, in particular, randomization is crucial   even for solving seemingly trivial problems. For example, in insertion-deletion streams, deterministic algorithms for outputting a single edge of the input graph require space $\Omega(n^2)$, where $n$ is the number of vertices in the input graph, while this problem can be solved via $\ell_0$-sampling \cite{jst11,cf14} and poly-logarithmic space when randomization is allowed. The predominant algorithmic technique for designing insertion-deletion streaming algorithms are  {\em linear sketches}, which are perfectly suited to the insertion-deletion model since they can naturally handle an arbitrary (unbounded) number of deletions (and insertions). It is even known that, under some conditions, linear sketches are universal for the class of insertion-deletion streaming algorithms in that the behavior of any such algorithm can be replicated by one that solely relies on the computation of linear sketches \cite{lnw14, ahlw16} (see also \cite{kp20} that further investigates said conditions). It is therefore not surprising that most previous work on streaming algorithms either considers the insertion-only setting with no deletions at all, or the insertion-deletion setting that allows for an unbounded/arbitrary number of deletions.

\smallskip
\noindent \textbf{Bounded-deletion Graph Streams.}
In this work, we take a more refined view and consider graph problems in {\em bounded-deletion streams}. In this setting, the input stream consists of an unrestricted number of edge insertions and at most $K$ edge deletions, for some integer $K$. We are interested in how the space requirements of streaming algorithms change as a function of $K$. 

Bounded-deletion, and, similarly, bounded-length, streams have been considered at least since the work of Cormode et al. \cite{cjmm17} in 2017 (see, e.g., \cite{cjmm17, bgw20, kp20, zmwaa21, agm12, zaam22, ejwy22, hkmms22,zaammr23,cgsv24}), who 
exploited a bound on the  input stream length for the design of a sampling-based algorithm for estimating the size of a largest matching in bounded-arboricity graphs. Besides this work, bounded-deletion graph streams appear also in Kallaugher and Price \cite{kp20}, who gave a graph problem for which linear sketching is exponentially harder than bounded-deletion streaming, and Chou et al. \cite{cgsv24}, who studied CSPs in this context, which in some cases can also be represented as graph problems. 

The bounded-deletion setting is well-motivated from a practical perspective. In all growing data sets, deletions are naturally less frequent than insertions since otherwise the data set would not grow in size. Regarding massive graphs, in social networks, new connections/friendships between entities are established much more often than deleted, and significantly more new hyperlinks are introduced into the web graph than deleted.

\smallskip
\noindent \textbf{The Maximal Matching Problem.} In this work, we study the \textsf{Maximal Matching} problem through the lens of bounded-deletion streams. A {\em matching} in a graph is a subset of non-adjacent edges, and a {\em maximal matching} is one that is inclusion-wise maximal, i.e., it cannot be extended by adding an edge outside the matching to it. Specifically, our goal is to understand the function $f(n,K)$ that describes the streaming space complexity of computing a maximal matching in an $n$-vertex graph as a function of the bound $K$ on the number of deletions. We will, however, also consider approximation algorithms to \textsf{Maximum Matching}, where the goal is to compute a matching of size at least $\alpha \cdot |M^*|$, with approximation factor $0 < \alpha \le 1$ and $M^*$ being a largest matching. It is well-known that a maximal matching constitutes a $\frac{1}{2}$-approximation to \textsf{Maximum Matching}.

The single-pass streaming space complexity of \textsf{Maximal Matching} is well-understood in both insertion-only streams and insertion-deletion streams with unrestricted deletions. In insertion-only streams, the simple \textsc{Greedy} matching algorithm, which greedily inserts every incoming edge into an initially empty matching if possible and constitutes the main building block of most streaming algorithms for matchings (e.g. \cite{fkmsz04,k18,kk20,alt21,kn21,kns23}), yields an $O(n \log n)$ space algorithm for \textsf{Maximal Matching}, and this is also tight. In insertion-deletion streams, Dark and Konrad \cite{dk20} showed an $\Omega(n^2)$ space lower bound for computing a \textsf{Maximal Matching}, strengthening a previous lower bound of $n^{2-o(1)}$~\cite{akl16} (see also \cite{k15}). In other words, there is essentially no better algorithm than storing the entire graph. The lower bound of Dark and Konrad requires $\Theta(n^2)$ deletions in the input stream.  When the number of deletions is restricted to be $K \in [n^2]$, the Dark and Konrad lower bound can be restated as showing that, for any $O(1)$-approximation to \textsf{Maximum Matching}, $\Omega(K)$ space is necessary. 

Summarizing the previous discussion, we know that $f(n,0) = \tilde{\Theta}(n)$, and that $f(n,n^2) = \Theta(n^2)$, but for arbitrary $K \in [n^2]$, known results only tell us that $f(n,K) = \Omega(n+K)$. On the algorithmic side, even when $K$ is just $O(n)$, no results are known for computing a maximal matching that utilize $o(n^2)$ space (i.e. do better than storing the entire graph). Our current state of knowledge, for instance, gives us the inequalities $\Omega(n) \le f(n,n) \le O(n^2)$, leaving a huge gap between the upper and lower bounds.
This raises a natural question: Does the required space $f(n,K)$ abruptly transitions to $\Omega(n^2)$ when $K$ is just $O(n)$, or is there is an algorithm that achieves a space complexity that smoothly interpolates between the $\tilde{O}(n)$ space for insertion-only streams, and $\Omega(n^2)$ space for unrestricted deletion streams? We show that, indeed, the space complexity smoothly interpolates between the two extremes.

\subsection{Our Results}
We resolve the space complexity of bounded-deletion streaming algorithms for \textsf{Maximal Matching} and show that $f(n,K)= \tilde{\Theta}(n \cdot \sqrt{K})$. We obtain our result  by designing a new space-efficient algorithm for bounded deletion streams, as well as by establishing a stronger new lower bound. 

\begin{theorem}\label{thm:ub}
There is a single-pass randomized  $\tilde{O}(n \cdot \sqrt{K})$ space streaming algorithm that, with high probability, outputs a maximal matching in any dynamic graph stream with at most $K$ deletions.
\end{theorem}
\newcounter{counterUB}  
\setcounter{counterUB}{\value{theorem}}

\begin{theorem}\label{thm:lb}
Any possibly randomized single-pass streaming algorithm that outputs a maximal matching with probability at least $2/3$ in dynamic graph streams with at most $K$ deletions requires $\Omega(n \cdot \sqrt{K})$ space.
\end{theorem}
\newcounter{counterLB}  
\setcounter{counterLB}{\value{theorem}}

It is worth noting that Theorem~\ref{thm:ub} implies that, up to poly-log factors, the $\Theta(n^2)$ deletions used in the lower bound by Dark and Konrad are necessary in order to obtain an $\Omega(n^2)$ space lower bound for \textsf{Maximal Matching}. Together these results show that the streaming space complexity of maximal matching increases smoothly as a function of the number of deletions, and there is no abrupt phase transition. On one extreme, when the number of deletions is $O(1)$, that is, when the deletions change the graph negligibly, the space complexity of $\tilde{\Theta}(n)$ is essentially the same as the space needed to store a maximal matching in an insertion-only stream. But then as the number of deletions reaches $\Omega(n^2)$, that is, when deletions can alter almost the entire graph, the space complexity rises to $\tilde{\Theta}(n^2)$, essentially the same as storing the entire graph.

We also observe that our work is the first that establishes a complete characterization of the space complexity of streaming algorithms for a graph problem as a function of the number of edge deletions $K$. 

We also show that randomization is crucial to achieving the space complexity given in Theorem~\ref{thm:ub}, as deterministic algorithms for maximal matching necessarily require $\Theta(n \cdot K)$ space. 

\begin{theorem}
\label{thm:det_maximal}
There is a single-pass streaming algorithm that uses $\tilde{O}(n \cdot K)$ space and outputs a maximal matching in any dynamic graph stream with at most $K$ deletions. Moreover, any deterministic algorithm for \textsf{Maximal Matching} requires $\Omega(n \cdot K)$ space.
\end{theorem}
\newcounter{counterDET}  
\setcounter{counterDET}{\value{theorem}}

Finally, we show that, unlike in unrestricted dynamic graph streams, in the bounded-deletion model, the space complexity of \textsf{Maximal Matching} behaves fundamentally differently to the space complexity of computing an $O(1)$-approximation to \textsf{Maximum Matching}. Let $g_c(n,K)$ be the streaming space complexity of computing a $c$-approximation to \textsf{Maximum Matching} in an $n$-vertex graph when the number of deletions is bounded by $K$. Then we know that for any constant $c > 2$, $g_c(n,0) = \tilde{\Theta}(n)=f(n,0)$, and that $g_c(n,n^2) = \Theta(n^2)=f(n,n^2)$. Furthermore, for arbitrary $K \in [n^2]$, we know that $g_c(n,K) = \Omega(n+K)$. We show that in a sharp contrast to the \textsf{Maximal Matching} problem, there is an algorithm that achieves the space complexity of $\tilde{O}(n+K)$, that is, $g_c(n,K) = \tilde{\Theta}(n+K)$, and this is achieved by a deterministic algorithm.

\begin{theorem}\label{thm:ub2}
For any $\epsilon > 0$, there is a deterministic single-pass streaming algorithm that uses $O((n + K/\epsilon) \cdot \log n)$ space and outputs a $(2+\epsilon)$-approximation to \textsf{Maximum Matching} in any dynamic graph stream with at most $K$ deletions.
\end{theorem}
\newcounter{counterUB2}  
\setcounter{counterUB2}{\value{theorem}}


\subsection{Techniques}
\label{subsec:Techniques}
\subparagraph{Bounded-deletions  \textsf{Maximal Matching} Algorithm.}
We start by motivating and explaining the main ideas underlying our main algorithmic result, namely, Theorem~\ref{thm:ub}. We start by observing that in absence of any deletions, one can simply store a maximal matching to solve the problem. On the other hand, when deletions are unbounded, we can simply store $\Theta(n^2)$ 
$\ell_0$-samplers to solve the problem by recreating the surviving graph. Our algorithmic approach below is based on a new data structure, called {\em hierarchical maximal matching}, which in conjunction with a hierarchical approach for storing $\ell_0$-samplers, allows us to design an algorithm whose space complexity smoothly interpolates between these two extremes.

Let $I$ denote the edges inserted during the stream, and let $D$ denote the edges deleted in the stream. Observe that $I$ and $D$ may be multisets since an edge can be inserted, subsequently deleted, and then reinserted again, and so on. Together, these sets define the graph $G(V,E)$ revealed by the dynamic stream where $E = I \setminus D$.

 Our algorithm has two phases where in the first phase, we build a {\em hierarchical} collection of maximal matchings using only edges in $I$ along with a data structure $\D$ for {\em matching repair} to account for deletions. In the second phase, we recover a maximal matching by starting with the least damaged maximal matching in our hierarchical collection, and repairing it using the edges stored in the data structure $\D$. We now explain this approach in more detail.

Suppose that we create a sequence of hierarchical maximal matchings, say, $M_1, M_2, ..., M_{L}$, using only the edges inserted in the stream (that is, the multiset $I$), ignoring all deletions. Specifically, we start by initializing $M_1, M_2, ..., M_L$ to be $\emptyset$. Now whenever an edge, say $(x,y)$, is inserted, we first try adding it to matching $M_1$. If one of $x$ or $y$ is already matched, then we try adding this edge $(x,y)$ to $M_2$, and continue in this manner. If we are unsuccessful all the way up to $M_L$, then we simply discard this edge. It is clear that this hierarchical collection can be implemented as a streaming algorithm using only $\tilde{O}(n \cdot L)$ space, since each matching can have only $O(n)$ edges. Now suppose there is an index $\ell \in [L]$ such that {\em none} of the edges in $M_{\ell}$ are deleted, that is, $M_{\ell} \cap D = \emptyset$. Then we can recover a maximal matching $M$ in the graph $G(V,E)$ as follows. We initialize $M = M_{\ell}$, and then greedily add edges in $(M_1 \cup M_2 \cup ... M_{\ell-1}) \setminus D$ to $M$ so that $M$ is a maximal matching w.r.t. edges in $(M_1 \cup M_2 \cup \dots \cup M_{\ell}) \setminus D$. We now claim that $M$ must be maximal in $G(V,E)$. Suppose not, then there is an edge $(u,v) \in E$, such that neither $u$ nor $v$ are matched in $M$. But then the edge $(u,v)$ must not be present in any of $(M_1 \cup M_2 \cup \dots \cup M_{\ell}) \setminus D$. This means that when the edge $(u,v)$ arrived in the stream, at least one of $u$ or $v$ must have been matched in $M_{\ell}$. This now is a contradiction to
our assumption, $M_{\ell} \cap D = \emptyset$, and hence it cannot be that both $u$ and $v$ are unmatched in $M$.
Thus $M$ is maximal with respect to $E = I \setminus D$.

Of course, the only way to ensure that there exists some index $\ell \in [L]$ such that {\em none} of the edges in $M_{\ell}$ are deleted is to set $L = \Omega(K)$, and doing so immediately yields our deterministic algorithm that uses space $\tilde{O}(n \cdot K)$.  
To obtain an $\tilde{O}(n \cdot \sqrt{K})$ space algorithm, we instead set $L = \sqrt{K}$, and observe that this ensures that there exists an index $\ell \in [L]$ such that at most $\sqrt{K}$ edges in $M_{\ell}$ are deleted. The second phase of the algorithm now starts on the task of repairing $M_{\ell}$ to be a maximal matching. Suppose an edge $(u,w)$ is deleted from $M_{\ell}$. We would like to see if the edges in $E \setminus (M_1 \cup M_2 \cup ... M_{\ell})$ can match $u$ and/or $w$ again. Let us focus on vertex $u$. Our data structure $\D$ will store $\Theta(\log^3 n)$ $\ell_0$-samplers for each vertex $x \in V$, with each sampler sampling uniformly at random from the incident edges on $x$. We open these $\ell_0$-samplers for $u$, and if we find an edge $(u,v)$ such that the vertex $v$ is unmatched, we add it to $M_{\ell}$. Otherwise, there are two possibilities: (a) the degree of $u$ is $O(\log^2 n)$ and we have recovered all incident edges on $u$, or (b) the degree of $u$ is $\Omega(\log^2 n)$, and we recover at least $\Omega(\log^2 n)$ distinct neighbors of $u$, all of whom are matched in $M_{\ell }$. In case (a), we are immediately in good shape because we have recovered all edges incident on $u$, and we can use them to match $u$ at the end if one of the neighbors of vertex $u$ remains unmatched. The more interesting situation is case (b) where on the one hand, the recovered information is not sufficient to repair $u$. On the other hand, we cannot rule out the possibility that $u$ could have been matched, if only we had allocated more space in $\D$ to recover additional edges incident on $u$. We next describe how we eliminate this uncertainity.

We create a collection $V_1, V_2, ..., V_R$ of subsets of $V$ where the set $V_i$ is a random subset of $V$ of size $n/\log^i n$, and $R = \Theta(\log n / \log \log n)$. For each vertex $w \in V_i$, our data structure $\D$ stores $\Theta(\log^{i+3} n)$ 
$\ell_0$-samplers. It is worth noting that while we are storing more and more $\ell_0$-samplers per vertex as $i$ increases, the total space used by vertices in $V_i$ remains $\tilde{O}(n)$ as the size of $V_i$ shrinks proportionately. 
We now return to repairing vertex $u$ that ended in case (b) above: we recovered a set $N(u)$ of $\Omega(\log^2 n)$ neighbors of $u$, such that each vertex in $N(u)$ is currently matched. Since every vertex is contained in $V_1$ with probability $1/\log n$, it follows that, with high probability, there must be a vertex $v_1 \in N(u)$ such that $v_1$'s mate in $M_{\ell}$, say $u_1$, belongs to $V_1$. We now add the edge $(u,v_1)$ to $M_{\ell}$, thereby matching vertex $u$, but now creating a new unmatched vertex, namely, $u_1$. It may seem as if we have not made any progress and the vertices $u$ and $u_1$ have simply traded places. But in fact we have made {\em some progress}. Since $u_1 \in V_1$, compared to vertex $u$, our data structure $\D$ stores $\Theta(\log n)$ times {\em more information} about edges incident on $u_1$, better positioning us to find a mate for $u_1$. We now repeat the above process for vertex $u_1$, either successfully matching it to an unmatched vertex, or recovering all incident edges on $u_1$, or finding a matched vertex $v_2 \in N(u_1)$ such that $v_2$'s mate in $M_{\ell }$, say $u_2$, belongs to $V_2$. The repair successfully terminates if either of the first two events occurs. Otherwise, we now add the edge $(u_1,v_2)$ to $M_{\ell}$, thereby matching vertex $u_1$, but now creating a new unmatched vertex in $M_{\ell }$, namely, $u_2$. We then continue this process from $u_2$. The repair process is guaranteed to successfully terminate when we reach $V_R$ since for each vertex in the final set $V_R$, our data structure $\D$ stores $\Theta(n \log n)$ $\ell_0$-samplers each, enough to recover their entire neighborhoods. 

To summarize, using $\tilde{O}(n)$ space, the data structure $\D$ provides a mechanism to repair an unmatched vertex $u$ in $M_{\ell}$ due to edge deletions. Since $M_{\ell}$ can have up to $\sqrt{K}$ edge deletions, to repair all of them, the data structure $\D$ independently replicates the above strategy $O(\sqrt{K})$ times to repair all deletions in $M_{\ell }$. The overall space used by the algorithm is thus $\tilde{O}(n \cdot \sqrt{K})$ for storing the edges in the hierarchical matching, and another $\tilde{O}(n \cdot  \sqrt{K})$ space for the repair data structure $\D$, giving us the desired space bound of $\tilde{O}(n \cdot \sqrt{K})$.

Finally, it is worth underlining that the computations of the hierarchical matching and the matching repair data structures are independent. Furthermore, the matching repair data structure is computed by a linear sketching algorithm, as it is usual in the insertion-deletion setting, and the hierarchical matching data structure is computed by a (non-linear) deterministic sketch, i.e., a Greedy algorithm, as is it typical in the insertion-only setting.

\subparagraph{Space Lower Bound for Bounded-deletions  \textsf{Maximal Matching}.}
We now explain the ideas behind our main lower bound result, established in Theorem~\ref{thm:lb}.
Our lower bound is best understood as an extension of the tight $\Omega(n^2 / \alpha^3)$ space lower bound by Dark and Konrad \cite{dk20} for one-pass insertion-deletion streaming algorithms that compute an $\alpha$-approximation to \textsf{Maximum Matching}\footnote{\cite{as22} gives an algorithm that achieves this space bound up to constant factors, see also \cite{k15,akl16,ccehmmv16}.}. We will denote this lower bound as the DK20 lower bound in the following. 
Due to the well-known fact that any maximal matching is at least half the size of a largest matching, DK20 immediately yields an $\Omega(n^2)$ space lower bound for \textsf{Maximal Matching}. 

In the following, we will treat the DK20 lower bound as if it was established for  \textsf{Maximal Matching}. DK20 is proved in the one-way two-party communication setting. In this setting, the first party, denoted Alice, holds the edge set $E$ of a bipartite graph $G=(A, B, E)$, and the second party, denoted Bob, holds edge deletions $D \subseteq E$. Alice sends a message to Bob, and, upon receipt, Bob is required to output a maximal matching in the graph $G' = (A, B, E \setminus D)$, i.e., Alice's input graph with Bob's deletions applied. Since Alice and Bob can simulate the execution of an insertion-deletion streaming algorithm for \textsf{Maximal Matching} on their input, by forwarding the memory state of the algorithm from Alice to Bob in form of a message, a lower bound on the size of the message used in the communication setting therefore also constitutes a lower bound on the memory required by such algorithms. 

In DK20, Alice holds a bipartite random graph $G=(A, B, E)$ with $|A| = |B| = n$ so that every edge is inserted into the graph with probability $\frac{1}{2}$. Bob holds subsets $A' \subseteq A$ and $B' \subseteq B$, with $|A'| = |B'| = \frac{4}{5} n$, and inserts deletions into the vertex-induced subgraph $G[A' \cup B']$ such that, after the deletions are applied, the remaining edges in $G[A' \cup B']$ form a large matching $M$. It is proved that recovering a constant fraction of the edges of $M$ -- a task that a protocol for \textsf{Maximal Matching} necessarily must achieve -- requires space $\Omega(n^2)$. On a technical level, DK20 give a sophisticated reduction to the well-known one-way two-party \textsf{Augmented-Index} communication problem. In \textsf{Augmented-Index}, Alice holds a bitstring $X \in \{0, 1\}^n$, and Bob holds a index $J \in [n]$ as well as the suffix $X[J+1, n]$. Alice sends a message to Bob who is tasked with reporting the bit $X[J]$. It is well-known that solving \textsf{Augmented-Index} with probability bounded away from $\frac{1}{2}$ requires a message of size $\Omega(n)$. In DK20, the reduction is such that the bits $X$ of an \textsf{Augmented-Index} instance correspond to edges in the random graph with the property that, with constant probability, the bit $X[J]$ corresponds to an edge in the matching $M$. This construction is such that the mapping between $X[J]$ and the edges in $M$ is random and, most importantly, unknown to the underlying protocol. Hence, a protocol that  reports a constant fraction of the edges of $M$ will therefore report the edge that corresponds to bit $X[J]$ with constant probability, which then solves \textsf{Augmented-Index}.

Bob holds $\Theta(n^2)$ deletions in the DK20 construction. In order to decrease the number of deletions to $K$, we proceed as follows. In our lower bound construction, Alice holds a graph $G=G_1 \ \dot{\cup} \ G_2 \  \dot{\cup} \ \dots \  \dot{\cup} \ G_s$, which is a vertex-disjoint union of $s$ graphs $(G_i)_{1 \le i \le s}$. Each graph $G_i$ has $\sqrt{K}$ vertices and constitutes a scaled-down version of Alice's input graph in the DK20 construction, namely a bipartite random graph with edge probability $\frac{1}{2}$. Bob holds an index $I \in \{1, \dots, s \}$, which identifies one of the graphs $G_I$. Furthermore, Bob holds the counterpart to Bob's input in the DK20 construction to graph $G_I$, i.e., edge deletions that apply to $G_I$. Bob leaves all the other graphs $G_j$, with $j \neq I$, untouched. 

We establish a direct sum argument to show that this problem requires a large message size. The key insight is that, since Alice does not know the index $I$, the message sent from Alice to Bob must contain sufficient information so that Bob can output a maximal matching no matter from which graph $G_i$ Bob has deleted some edges. In other words, Alice and Bob must be able to solve $s$ independent copies of the DK20 lower bound instance. Since each graph $G_i$ has $\sqrt{K}$ vertices, DK20 implies that $\Omega(K)$ bits are required for solving one copy. Hence, overall, space $\Omega(K \cdot s)$ is required. Last, to make sure that the final graph has $n$ vertices, we need to set $s = \Theta(\frac{n}{\sqrt{K}})$, which delivers the claimed $\Omega(n \cdot \sqrt{K})$ space lower bound.

To implement this approach, we first define a generalization of \textsf{Augmented-Index} denoted \textsf{Embedded-Augmented-Index}, where Alice holds $s$ binary strings $X_1, \dots, X_s \in \{0, 1\}^t$, Bob holds two indices $I \in [s]$ and $J \in [t]$ as well as the suffix $X_I[J+1, t]$, and the objective for Bob is to output the bit $X_I[J]$.
Using by now standard information-theoretic arguments, we show that this problem requires $\Omega(t \cdot s)$ bits of communication. Next, following the proof outline of DK20, we show that a bounded-deletion streaming algorithm can be used to solve \textsf{Embedded-Augmented-Index}, which completes the proof.


\subparagraph{Lower Bound for Deterministic Algorithms for Bounded Deletions \textsf{Maximal Matching}.}


Our lower bound for deterministic algorithms works with a family of bipartite graphs $G=(A, B, E)$ that have the property that, even when any $K$ edges $D \subseteq E$ are deleted from $G$, then still every maximal matching in $G - D$ matches all $A$-vertices. 
The family of graphs $\mathcal{G}_K(n)$ is obtained as follows. Given the deletion budget $K$, define $\mathcal{H}_K$ to be the family of bipartite graphs $H=(A, B, E)$ with $|A| = K, |B| = 3K$, and the degree of every $A$-vertex is $2K$. Then, $\mathcal{G}_K(n)$ is the family of graphs consisting of the disjoint union of any $n/(4K)$ graphs from $\mathcal{H}_K$. 

Similar to our lower bound for randomized algorithms, we prove our lower bound in the one-way two-party communication setting. Alice holds a graph $G \in \mathcal{G}_K(n)$ as input, and Bob holds up to $K$ edge deletions $D \subseteq E(G)$. We now claim that a protocol $\pi$ for \textsf{Maximal Matching} that outputs a maximal matching in the graph $G - D$ allows Bob to learn $K+1$ incident edges on every $A$-vertex. Once this claim is established, we finalize the proof by observing that, if the message from Alice to Bob was of size $o(n \cdot K)$, then $\pi$ also constitutes an encoding of these overall $\frac{n(K+1)}{4}$ edges`using only $o(n \cdot K)$ bits. This in turn can be used to encode the graph class $\mathcal{G}_K(n)$ with fewer bits than dictated by the information-theoretic threshold - a contradiction.

To see that Bob can learn $K+1$ edges incident on every $A$-vertex in input graph $G \in \mathcal{G}_K(n)$, Bob proceeds as follows. Let $a \in A$ be any vertex and $\pi$ the message received from Alice. Bob then completes the protocol without introducing any deletions, recovering a maximal matching $M_1$ that, as discussed above, necessarily matches the vertex $a$. Let $e_1$ be the edge incident on $a$ in $M$. Next, Bob completes another run of the protocol starting with message $\pi$ and feeding the edge $e_1$ as edge deletion into the protocol. In doing so, Bob obtains another maximal matching $M_2$ that necessarily matches $a$, thereby recovering a second edge $e_2$ incident on $a$. Repeating this process, Bob learns more and more edges incident on $a$, feeding these edges as deletions into the protocol in the next simulation. Overall, Bob can repeat this process $K$ times, thereby exhausting the deletion budget, which allows him to learn $K+1$ edges incident on $a$. Finally, this process can be repeated for every $a \in A$, which completes the argument.

\subparagraph{Bounded-deletions $(2+o(1))$-Approximate Matching Algorithm.}
We now briefly describe the main idea behind the algorithmic result stated in Theorem~\ref{thm:ub2}. Once again, similar to Theorem~\ref{thm:ub}, our algorithm utilizes the hierarchical matching data structure, but with some crucial modifications.

Specifically, this time we implement a {\em budgeted version} of the hierarchical matching data structure on the insertions part of the stream, where instead of having a fixed number of levels, we now have an overall budget $B$ on the total number of edges stored in the hierarchical matching data structure. The data structure maintains a lexicographic maximality property whereby an edge insertion that may cause the number of stored edges to exceed the budget $B$ triggers removal of an edge from the {\em current last level} of the data structure. The number of levels in the data structure forms a {\em bitonic} sequence over time, in that, it may keep increasing as long as the space budget $B$ has not been reached, but once that happens, the number of levels steadily decreases as enforcing the lexicographic property can trigger many 
deletions from the last levels of the hierarchical matching as the remainder of the insertion stream is processed.
 
An immediate consequence of this lexicographic property is that if upon termination, we have a sequence of hierarchical matchings, say, $M_1, M_2, ..., M_{L}$, then $M_1, M_2, ..., M_{L-1}$ must necessarily constitute a $(L-1)$-level hierarchical maximal matching data structure. We also separately store all deletions that appear in the stream. Now by setting our budget $B$ to be $\tilde{\Theta}(n+K)$, we can show that when we apply deletions, one of the $(L-1)$ maximal matchings, say $M_i$, necessarily loses at most an $o(1)$-fraction of its edges. We can then extend $M_i$ to be an ``almost'' maximal matching $M$ in the overall stream, by greedily inserting surviving edges in $M_1, M_2, ..., M_{i-1}$ into it. Here the ``almost'' refers to the fact that unlike the maximal matching algorithm of Theorem~\ref{thm:ub}, we do not need to repair the missing $o(1)$-fraction edges in $M_i$ to ensure maximality. The final matching $M$ is easily shown to be a $(2+o(1))$-approximate matching, giving us the desired result.

\subsection{Outline}
In Section~\ref{sec:prelim}, we give notation and the definition of bounded-deletion graph streams as well as some key information-theoretic facts that we use to prove our lower bound. 
We then present our bounded-deletion algorithm for \textsf{Maximal Matching} in Section~\ref{sec:alg}, and give our matching lower bound in Section~\ref{sec:lb}. Our results on deterministic algorithms are then presented in Sections~\ref{sec:det} and \ref{sec:maximum-matching}, i.e., we settle the space complexity of deterministic bounded-deltion algorithms for \textsf{Maximal Matching} in Section~\ref{sec:det}, and we give our $\tilde{O}( n + K)$ space algorithm for $O(1)$-approximation to \textsf{Maximum Matching} in Section~\ref{sec:maximum-matching}. We conclude with some directions for future research in Section~\ref{sec:conclusion}.

\section{Preliminaries}\label{sec:prelim}
\subparagraph{Notation.} For an integer $k$ we write $[k] := \{1, 2, \dots, k \}$. For a graph $G=(V, E)$, we denote by $N(v)$ the neighborhood of a vertex $v \in V$.

\subparagraph{Bounded-deletion Graph Streams.} Given an integer $K$, in the {\em bounded-deletion graph stream setting}, the input stream consists of a sequence of edge insertions and at most $K$ edge deletions, which together make up the edge set of a graph $G=(V, E)$. An edge deletion can only occur if the edge was previously inserted. Observe that it is possible that an edge is introduced, subsequently deleted, and then introduced again, and so on. We assume that every prefix of the input stream describes a simple graph, i.e., a previously inserted edge can only be inserted again if it was deleted in the meantime.  Observe, however, that the substream of edge insertions $I$ may constitute a multiset, and, when regarded as a graph, then $I$ yields a multigraph. The same holds for the substream of edge deletions $D$. In the following, we will write $E= I \setminus D$ to denote the edges of the final graph described by the input stream.  We will see that our algorithm considers the substream of edge insertions and deals with edge multiplicities in a rather natural way. 

We also assume that the parameter $K$ is known in advance. This is a necessary assumption since, if $K$ was not known, any algorithm that uses space $o(n^2)$ could not produce a maximal matching in case $K= \Theta(n^2)$, as demonstrated by the Dark and Konrad lower bound \cite{dk20}. As it is standard, for simplicity, we also assume that algorithms know the vertex set $V$ of the graph described by the input stream in advance.

Our algorithm makes use of $\ell_0$-sampling. Given an insertion-deletion stream that describes a vector on $m$-coordinates, $\ell_0$-sampling refers to the process of sampling a uniform random non-zero coordinate of this vector. We will use the $\ell_0$-samplers of Jowhari et al. \cite{jst11}:
\begin{theorem}[Jowhari et al. \cite{jst11}]
    There exists a $\ell_0$-sampler that uses $O(\log^2 m \log(1/\delta))$ bits of space and outputs a uniform random non-zero coordinate $i \in [m]$ with probability at least $1 - \delta$.
\end{theorem}
Observe that an edge stream of an $n$-vertex graph  describes a $\Theta(n^2)$-dimensional vector. $\ell_0$-sampling on a graph stream therefore produces a uniform random edge of the input graph with space $O(\log^2 n \cdot \log(1/\delta))$. We will use $\ell_0$-sampling to sample uniform random edges incident to a given vertex $v \in V$, by feeding only edge insertions and deletions incident on $v$ into the sampler.

\subparagraph{Information-theoretic Arguments} 
Our lower bounds are proved using information-theoretic arguments. We refer the reader to \cite{ct06} for an excellent introduction to information theory.

For jointly distributed random variables $(A,B,C) \sim \mu$, we denote by $H_{\mu}(A)$ and $H_{\mu}(A \ | \ B)$ the Shannon entropy of $A$ and the conditional entropy of $A$ conditioned on $B$, respectively. Furthermore, we denote by $I_{\mu}(A \ : \ B)$ and $I_{\mu}(A \ : \ B \ | \ C)$ the mutual information between $A$ and $B$ and the conditional mutual information between $A$ and $B$ conditioned on $C$, respectively. 

We will use the following key properties of entropy and mutual information: (let $(A,B,C,D) \sim \mu$ be jointly distributed random variables)

\begin{enumerate} 
 \item Definition of conditional mutual information: $I_{\mu}(A \ : \ B \ | \ C) = H_{\mu}(A \ | \ C) - H_{\mu}(A \ | \ B,C)$,

 \item Definition of conditional mutual information via expectation: $$\displaystyle I_{\mu}(A \ : \ B \ | \ C) = \Exp_{c \gets C} I_{\mu}(A \ : \ B \ | \ C = c) \ ,$$

 \item Chain rule for mutual information: $I_{\mu}(A \ : \ B,C \ | \ D) = I_{\mu}(A \ : \ B \ | \ D) + I_{\mu}(A \ : \ C \ | \ B,D)$,

 \item Conditioning reduces entropy: $H_{\mu}(A \ | \ B) \le H_{\mu}(A)$,

 \item Independence of an event: Let $E$ be an event that is independent of $A,B,C$. Then,  $$I_{\mu}(A \ : \ B \ | \ C, E) = I_{\mu}(A \ : \ B \ | \ C) \ . $$
\end{enumerate}

Furthermore, we will also use (a consequence of) Fano's inequality as stated in Lemma~\ref{lem:fano} and a standard property about mutual information as stated in Lemma~\ref{lem:property-mi}.
\begin{lemma}[Fano's Inequality \cite{ct06}]\label{lem:fano}
Let $(X,Y, \tilde{X}) \sim \mu$ be jointly distributed random variables such that $X$ is a binary random variable that takes on values in $\{0,1\}$, $\tilde{X}$  is an estimator of $X$ that also takes on values in $\{0, 1\}$ such that $\tilde{X}$ is a random function of a variable $Y$. Then, we have:
 \begin{align*}
     H_{\mu}(X \ | \ Y) \le H_2(\Pr[\tilde{X} \neq X]) \ , 
 \end{align*}
 where $H_2$ denotes the binary entropy function.
\end{lemma}
We refer the reader to \cite{ct06} for a more general version of Fano's Inequality.

\begin{lemma} \label{lem:property-mi}
    Let $(A,B,C,D) \sim \mu$ be jointly distributed random variables such that $A$ and $D$ are independent conditioned on $C$. Then:
    \begin{align*}
        I_{\mu}(A \ : \ B \ | \ C) \le I_{\mu}(A \ : \ B \ | \ C, D) \ . 
    \end{align*}
\end{lemma}
A proof of the previous statement can be found, for example, in \cite{akl16} (Claim~2.3.).

\section{An $\tilde{O}(n \sqrt{K})$ Space Algorithm for \textsf{Maximal Matching}} \label{sec:alg}

\subsection{Description of the Algorithm}

We start by briefly summarizing the main steps of our algorithm. We refer the reader to the detailed overview in Section~\ref{subsec:Techniques} for detailed intuition underlying the steps of our algorithm.

Let $G=(V,E)$ be the underlying graph that is revealed as a sequence $I$ of edge insertions and a sequence $D$ of deletions, that is, $E = I \setminus D$. 
In the first phase, our algorithm  builds a sequence of hierarchical matchings, say, $M_1, M_2, ..., M_{L}$, using only $I$, the edges inserted in the stream, ignoring all deletions, where $L = \sqrt{K}$. 
Since there are at most $K$ deletions, there exists an index $\ell \in [1..L]$ such that at most $\sqrt{K}$ edges in $M_{\ell}$ are deleted in the set $D$. The second phase of the algorithm is then focused on restoring the maximality of the matching $M_{\ell}$ after the deletions are applied. 

This is achieved as follows. 
We create a collection $V_0, V_1, V_2, \dots, V_R$ of subsets of $V$ where the set $V_i$ is a random subset of $V$ of size $n/\log^i n$, and $R = \Theta(\log n / \log \log n)$. For each vertex $w \in V_i$, our data structure $\D$ stores $\Theta(\log^{i+3} n)$ $\ell_0$-samplers, sampling uniformly at random from the edges incident on $w$. Now suppose vertex $u$ used to be matched in $M_{\ell}$ via an edge $(u,v)$, and this edge is deleted by $D$, then the repair process starts by opening the $\ell_0$-samplers of $u$. This results in one of the following three events with high probability: 

\begin{itemize}
    \item[(i)] either $u$ can be matched to a previously unmatched vertex, or 
    \item [(ii)] $u$ recovers all its incident edges in $G$, or 
    \item [(iii)] $u$ can be matched to an already matched vertex $v_1$ whose current partner $u_1$ is in the set $V_1$.
\end{itemize}
 Events (i) and (ii) are clearly (successful) termination events for $u$'s repair as we have either re-matched $u$, or we have all incident edges on $u$ available for us to try to rematch $u$ -- we, however, only do this after the repair process of all vertices has terminated. In case (iii), the repair process now continues from $u_1$, exploiting the fact that our data structure stores more $\ell_0$-samplers for each vertex in $V_1$. We will once again encounter one of the above three events, and in case of the third event, we will once again repair $u_1$ by unmatching one of its neighbors, and continue repair from a vertex, say $u_2 \in V_2$.  The process is guaranteed to terminate successfully when it reaches a vertex in $V_R$ since we store enough $\ell_0$-samplers to recover the entire neighborhood of each vertex in $V_R$. 
Finally, the algorithm replicates the data structure $\Theta(\sqrt{K})$ times to ensure the repair of every vertex in $M_{\ell}$ that is affected by deletions in $D$.

Our main algorithm is depicted as Algorithm~\ref{alg:maximal}, and it uses the \textsf{Hierachical-Greedy}() algorithm, which is stated as Algorithm~\ref{alg:hierachical-matching}, as a subroutine.

\begin{algorithm}
    \begin{algorithmic}[1]
        \REQUIRE Input stream of edge insertions (multi-edges allowed), integer parameter $L \ge 1$
        \STATE $M_1, M_2, \dots, M_{L} \gets \varnothing$
        \WHILE{stream not empty}
            \STATE Let $e$ be the next edge in the stream
            \STATE Let $\ell \ge 1$ be the smallest integer such that $M_{\ell} \cup \{e\}$ is a matching, if no such matching exists then let $\ell = -1$
            \IF{$\ell \neq -1$}
                \STATE $M_{\ell} \gets M_{\ell} \cup \{e\}$
            \ENDIF
        \ENDWHILE
        \RETURN $(M_1, M_2, \dots, M_{L})$
    \end{algorithmic}
    \caption{\textsf{Hierachical-Greedy($L$)}  \label{alg:hierachical-matching}}
\end{algorithm}

We now give some more implementation details of our main algorithm that are not covered by the pseudo-code.

In Line~\ref{line:apply-del} in the listing of Algorithm~\ref{alg:maximal}, we apply the deletions $D$ to the matchings of the hierachical matching. While this is straightforward if the substream $I$ of edge insertions does not contain any multiedges, some care needs to be taken if $I$ does contain multiedges. We apply the deletions in order as they arrived in the stream and bottom-up, i.e., we first apply deletions to $M_1$, then $M_2$, and so on. This is to ensure that deletions are matched to the relevant insertions.

In Line~\ref{line:entire-neighborhood} of our main algorithm, Algorithm~\ref{alg:maximal}, we evaluate whether the entire neighborhood of a vertex is contained in a set of $\ell_0$-samplers. This condition can, for example, be checked by comparing the number of different incident edges produced by the $\ell_0$-samplers to the degree of the vertex. Observe that, in insertion-deletion streams, the degree of a vertex can easily be computed by using a counter in $O(\log n)$ space.

\begin{algorithm}[h!]
    \begin{algorithmic}[1]
        \REQUIRE Stream of edge insertions and at most $K$ edge deletions making up a graph $G=(V, E)$; large constants $C,C'$  with $C > C'$ suitably chosen \vspace{0.15cm}
        \STATE $R \gets \frac{\log(n) + \log(C')}{\log \log n} - 2$ \COMMENT{Number of vertex levels}
        \STATE Let $V_0 = V$, and for $i \in [R]$, let $V_i \subseteq V$ be a random subset of size $C \cdot \frac{n}{\log^i(n)}$\vspace{0.3cm}
        
        \STATE \textbf{Input Stream Processing: Run in parallel}
        \STATE $\quad$ 1. $(M_1, \dots, M_{\sqrt{K}}) \gets $ \textsf{Hierachical-Greedy}($\sqrt{K}$) on substream of edge insertions
        \STATE $\quad$ 2. Store all edge deletions observed in the stream in variable $D$
        \STATE $\quad$ 3. For each $i \in [R] \cup \{0 \}$ and every $v \in V_i$, run $2 \cdot \sqrt{K}\log^{i+3}(n)$ $\ell_0$-samplers on the substream of edges incident on $v$; each sampler uses failure parameter $\delta = \frac{1}{n^4}$ \vspace{0.15cm}
        \STATE \textbf{Notation:} For any $i \in [R] \cup \{0 \}$, $v \in V_i$, and $j \in [\sqrt{K}]$ denote by $N_{i,j}(v)$ the $j$th $\frac{1}{2 \cdot \sqrt{K}}$-fraction of $v$'s level $i$ $\ell_0$-samplers \vspace{0.3cm}        
        \STATE \textbf{Post-processing:}
        \STATE Let $(M_1', M_2', \dots, M_L')$ be the matchings $(M_1, M_2, \dots, M_L)$ with the deletions $D$ applied \label{line:apply-del}
        \STATE Let $\ell$ be such that $|M_{\ell} \setminus M_{\ell}'| \le \sqrt{K}$, and let $V(M_{\ell} \setminus M_{\ell}') = \{u_1, u_2, \dots \}$ 
        \STATE $M \gets M_{\ell}'$ \label{line:del} \COMMENT{repair $M$ back to a maximal matching}
        \FOR[fix $u_j$]{$j \gets 1 \dots |V(M_{\ell} \setminus M_{\ell}')|$} \label{line:loop-repair}
            \IF{$u_j \in V(M)$} \label{line:accidental-fix}
                \STATE \textbf{continue} \COMMENT{$u_j$ was matched while fixing some $u_b$ with  $b < j$}
            \ENDIF
            \STATE $u \gets u_j$            
            \FOR[Iterate through the vertex levels]{$i = 0 \dots R$} \label{line:loop}
            \IF{$N_{i,j}(u)$ contains entire neighborhood of $u$} \label{line:entire-neighborhood}
                \STATE \textbf{continue} \COMMENT{$u$ will be dealt with in Line~\ref{line:final-fix}}
            \ELSIF{$N_{i,j}(u)$ contains a vertex $v$ that is not matched in $M$} 
                \STATE $M \gets M \cup \{(u,v) \}$
            \ELSE
                \STATE Let $v \in N_{i,j}(u)$ be such that its mate $u'$ in $M$ is such that $u' \in V_{i+1}$ (Lemma~\ref{lem:next-level} shows that such a vertex exists w.h.p.)
                \STATE $M \gets (M - (u', v)) \cup \{(u,v) \}$
                \STATE $u \gets u'$
            \ENDIF
        \ENDFOR \label{line:end-loop-repair}
        \STATE Greedily attempt to add all edges in $M_1', \dots, M_{\ell - 1}'$ as well as all edges recovered by the $\ell_0$-samplers to $M$ if possible \label{line:final-fix}
    \ENDFOR
    \RETURN $M$
    \end{algorithmic}
    \caption{Bounded-deletion Streaming Algorithm for \textsf{Maximal Matching}\label{alg:maximal}}
\end{algorithm}

\subsection{Analysis}
Before analyzing our algorithm, we point out that the description of the algorithm uses large constants $C,C'$ with $C > C'$ that are appropriately chosen. We will see that the analysis only imposes weak constraints on $C$ and $C'$ and that such constants are easy to pick.
We also assume that all our $\ell_0$-samplers succeed. Since we run the $\ell_0$-samplers with error parameter $\delta = \frac{1}{n^4}$, and there are less than $n^2$ such samplers, by the union bound, this is a high probability event.

The first key ingredient of our analysis is the fact that, in the absence of deletions, every matching $M_\ell$, $1 \le \ell \le L$, produced by \textsf{Hierachical-Greedy}() can easily be extended to a globally maximal matching by greedily adding edges of the matchings $M_1, \dots, M_{\ell-1}$ to it if possible.  
The next lemma captures this idea and combines it with a key insight that allows us to fix edge deletions: We do not need to immediately rematch vertices incident to a deleted edge as long as we know their entire neighborhoods. This lemma is key for establishing our algorithm's correctness.

\begin{lemma}\label{lem:key-1}
Let $G(V,E)$ be a graph where $E = I \setminus D$ such that $I$ is a multiset of edge insertions and $D \subseteq I$ is a multiset of edge deletions. Let $(M_1, M_2, \dots, M_{L})$ be the output of \textsf{Hierachical-Greedy}$(L)$, i.e., a sequence of hierarchical matchings constructed using the edges in $I$, processed in an arbitrary order. Let $F \subseteq E$ be a subset of edges. Then, for any $\ell \in [L]$, let $M$ be any matching such that, for every vertex $v$ matched in $M_{\ell}$, either (i) $v$ is also matched in $M$, or (ii) the entire neighborhood of $v$ is known, i.e., $N(v) \subseteq F$. 
Then $M$ can be extended to become a maximal matching of the graph $G(V,E)$ by simply greedily adding to it edges in $F \cup \left( (M_1 \cup M_2 \cup ... M_{\ell-1}) \setminus D \right)$ if possible.
\end{lemma}
\begin{proof}
    We prove this lemma by contradiction. To this end, assume that there exists an edge $(u,v) \in E \setminus M$ such that $M \cup \{ (u, v) \}$ is a matching. 
    
    We first argue by contradiction that $V(M_{\ell}) \cap \{u,v\} = \varnothing$. Indeed, suppose that this was not the case and, w.l.o.g., assume that $u \in V(M_{\ell})$. By the statement of the lemma, then either $u$ is matched in $M$ or the entire neighborhood of $u$ is known, i.e., $N(u) \subseteq F$, which implies that $(u, v) \in F$. In the first case, we immediately arrive at a contradiction since the edge $(u,v)$ could not be added to $M$ since $u$ is already matched. In the second case, observe that we attempted to greedily add the edges of $F$ to $M$, in particular, we already attempted to add the edge $(u,v)$ to $M$, which also yields a contradiction.

    Assume therefore that $V(M_{\ell}) \cap \{u,v\} = \varnothing$. This, however, implies that, when the edge $(u,v)$ arrived in the stream, it was not included in $M_{\ell}$ despite both endpoints being free in $M_{\ell}$. Suppose the edge $(u,v)$ occurs $\beta$ times in the stream $I$. It must then be the case that {\em all} $\beta$ copies of this edge in $I$ must have been added to one or more of the matchings $M_1, \dots, M_{\ell-1}$. On the other hand, since edge $(u,v) \in E \setminus M$, it means that the numbers of times $(u,v)$ is deleted in $D$ is strictly less than $\beta$. So at least one copy of the edge $(u,v)$ in $M_1, \dots, M_{\ell-1}$ remains undeleted, after we remove edges in $D$. Since we attempted to add the surviving edges of $M_1, \dots, M_{\ell-1}$ to $M$, we also arrive at a contradiction. 

   It follows then that the matching $M$ is indeed maximal in $G(V,E)$.

\end{proof}

The second key ingredient of our analysis is a {\em progress lemma} that shows that the fixing process yields the desired result. In more detail, we show that, in iteration $i$ of the loop in Line~\ref{line:loop}, when fixing any vertex $u_j \in (V(M_\ell) \cap V_i)$ that is currently unmatched but was matched in $M_{\ell}$, then we either (i) recover all of $u_j$'s neighbors, (ii) we are able to match $u_j$ directly to a yet unmatched vertex, or (iii) we can match $u_j$ to an already matched vertex $v$ such that $v$'s mate $u'$ in the current matching is contained in level $i+1$, i.e., $u' \in V_{i+1}$.

\begin{lemma}[Progress Lemma]\label{lem:next-level} \label{lem:progress}
 Let $M$ be any matching in $G=(V, E)$, where $G$ is the final graph described by the input stream, and let $u \in V_i \setminus V(M)$ be an unmatched vertex contained in level $i$. Then, for any $j$, with high probability, at least one of the following assertions is true:
 \begin{enumerate}
     \item $N_{i,j}(u)$ contains the entire neighborhood of $u$; \label{item:one}
     \item $N_{i,j}(u)$ contains a vertex $v$ such that $M \cup \{(u,v) \}$ is a matching; \label{item:two}
     \item $N_{i,j}(u)$ contains a vertex $v$ that is matched in $M$ to a vertex $u'$ such that $u' \in V_{i+1}$. \label{item:three}
 \end{enumerate}
\end{lemma}
\begin{proof}We first recall that $N_{i,j}(u)$ consists of $\log^{i+3}(n)$ $\ell_0$-samplers. 

    Suppose that Items~\ref{item:one} and \ref{item:two} are false. Since Item~\ref{item:one} is false, the degree of $u$ cannot be too small since otherwise the $\ell_0$-samplers would have picked up every single edge incident to $u$. We prove in  Lemma~\ref{lem:technical-1} that, w.h.p., the degree of $u$ is at least $\log^{i+2}(n)/C'$. Then, since Item~\ref{item:two} is false, all vertices produced by the $\ell_0$-samplers in $N_{i,j}(u)$ are matched in $M$. Denote by $U$ their mates in $M$. Then, $|U| \ge \log^{i+2}(n)/C'$. Then, the probability that none of the vertices in $U$ are contained in $V_{i+1}$ is:
    $$\frac{{n - |U| \choose |V_{i+1}|}}{    
    {{n \choose |V_{i+1}|}}}  \le \left(1 - \frac{|V_{i+1}|}{n} \right)^{|U|} \le \exp(-|V_{i+1}| \cdot |U| / n) = \exp(-\frac{C}{C'} \log n) \le \frac{1}{n^5} \ ,$$
where we used Lemma~\ref{lem:tech} stated in the appendix, then used the inequality $1+x \le \exp(x)$, and then assumed that $C$ and $C'$ are picked such that $\frac{C}{C'}$ is large enough. Hence, such a vertex exists with high probability, which completes the proof.

\end{proof}

\begin{lemma} \label{lem:technical-1}
 Let $v \in V_i$ be any vertex. Then, if $\deg(v) \le \log^{i+2}(n)/C'$ then $N(v) \subseteq N_{i,j}(v)$, for every $j$, with high probability.
\end{lemma}
\begin{proof}
    First, observe  that $N_{i,j}(v)$ consists of $\log^{i+3}(n)$ $\ell_0$-samplers for each vertex $v \in V_i$. Let $u \in N(v)$ be a vertex. Then, assuming that none of the $\ell_0$-samplers fail, the probability that the outcome of an $\ell_0$-sampler for $v$ is $u$ is at least $\frac{1}{\deg(v)} \ge \frac{C'}{\log^{i+2}(n)}$. Since all samplers operate independently, the probability that none of the samplers produce $u$ is at most:
$$\left( 1 -  \frac{C'}{\log^{i+2} n} \right)^{\log^{i+3}  n} \le \exp \left(- \frac{C'}{\log^{i+2} n} \cdot  \log^{i+3} n \right) = \exp \left( - C' \log n \right) \le \frac{1}{n^{10}} \ ,$$
using the inequality $1+x \le \exp(x)$ and by picking a sufficiently large value for $C'$. By a union bound over the vertices in $N(v)$, we obtain that all vertices in $N(v)$ are contained in the $\ell_0$-samplers with high probability. Last, by a union bound over all values of $j$, the result follows.
\end{proof}

We observe that the previous lemma implies the important property that, for every vertex $v \in V_{R}$ with $R = \frac{\log(n) + \log(C')}{\log \log n} - 2$ as in the algorithm, i.e., a vertex contained in the last level $R$, the entire neighborhood of $v$ is contained in $N_{R,j}(v)$ since
$\log^{R+2}(n) / C' \ge n \ . $

We are now ready to prove our main result, i.e., Theorem~\ref{thm:ub}.

\newcounter{thmsaved}
\setcounter{thmsaved}{\value{theorem}}
\setcounter{theorem}{\value{counterUB}}
\addtocounter{theorem}{-1}

\begin{theorem}
 Algorithm~\ref{alg:maximal} is a $\tilde{O}(n \cdot \sqrt{K})$ space streaming algorithm for \textsf{Maximal Matching}, where $K$ is an upper bound on the number of deletions in the stream, and succeeds with high probability.
\end{theorem}
\setcounter{theorem}{\value{thmsaved}} 

\begin{proof} 
    \noindent We first address the space complexity of the algorithm and then establish its correctness. 

    \vspace{0.15cm}
    \noindent \textbf{Space Complexity.} 
    The algorithm stores the $\sqrt{K}$ matchings $M_1, \dots, M_{\sqrt{K}}$, which requires space $\tilde{O}(n \sqrt{K})$. In addition, the algorithm stores various $\ell_0$-samplers. The number of these samplers is bounded by 
    $$\sum_{i=0}^R |V_i| \cdot 2 \sqrt{K} \log^{i+3}(n) \le  \sum_{i=0}^R C \cdot \frac{n}{\log^i n} \cdot 2 \sqrt{K} \log^{i+3}(n) = n \cdot \sqrt{K} \cdot \log^3 n \cdot R = \tilde{O}(n \sqrt{K}) \ , $$
    and, since each $\ell_0$-sampler uses space $\tilde{O}(1)$, the space bound follows.
    
    \vspace{0.15cm}
    \noindent \textbf{Correctness.}
    To establish correctness, we will argue that after the loop in Line~\ref{line:loop-repair}, which ends in Line~\ref{line:end-loop-repair}, but before Line~\ref{line:final-fix} is executed, the matching $M$ together with the set of $\ell_0$-samplers fulfill the premises of Lemma~\ref{lem:key-1}. Invoking the lemma then establishes that executing Line~\ref{line:final-fix} turns $M$ into a globally maximum matching, which  completes the proof.

    To argue that $M$ fulfills the premises of Lemma~\ref{lem:key-1} after Line~\ref{line:end-loop-repair} (but before Line~\ref{line:final-fix}), we need to ensure that, for every vertex $u \in V(M_\ell)$, either $u$ is matched in $M$ or the entire neighborhood of $u$ is contained in our set of $\ell_0$-samplers. 
    In the following, we say that a vertex $u \in V(M_{\ell})$ fulfills property $P$ if it is either matched in $M$ or its entire neighborhood is contained in our $\ell_0$-samplers.
    
    Observe that when entering the loop in Line~\ref{line:loop-repair}, the vertices $V(M_{\ell} \setminus  M_{\ell}')$ are exactly those vertices that violate property $P$. We will now argue that each iteration of this for-loop fixes one of these vertices and, in particular, it does not introduce any new vertices that may violate this property. This then completes the proof.

    Consider thus the iteration $j$ that fixes vertex $u_j$. First, it may happen that $u_j$ was matched when fixing a vertex $u_{j'}$, for some $j' < j$. This is checked in Line~\ref{line:accidental-fix}, and if this happens then we are done. Otherwise, we enter the loop in Line~\ref{line:loop}. The vertex $u$ is initialized as the vertex $u_j$ that we attempt to fix. The loop satisfies the invariant that, in iteration $i$, we are guaranteed that $u \in V_i$. This is clearly the case for the first iteration $i=0$ since $V_0 = V$. To see that this holds throughout, we will argue that an iteration of the loop always establishes property $P$ for the current vertex $u$. Fixing vertex $u$ may however come at the cost of introducing a new violation of property $P$ for a previously matched vertex $u'$. This vertex $u'$, however, is then necessarily  contained in the vertex set $V_{i+1}$. Since at the end of the loop, we set $u \gets u'$, we therefore see that in the subsequent iteration $i+1$, the vertex $u$ is now contained in $V_{i+1}$, as desired. 

    Now, to see that $u$ is always fixed, we invoke our progress lemma, Lemma~\ref{lem:progress}, which states that either the entire neighborhood of $u$ is contained in our $\ell_0$-samplers (case 1) or $u$ is rematched (cases 2 and 3), which implies that $u$ now satisfies property $P$. Case 3 matches $u$ to a previously matched vertex $v$ whose mate $u'$ becomes unmatched. However, as proved in Lemma~\ref{lem:progress}, $u'$ is then necessarily contained in $V_{i+1}$. 

    Last, to see that in the final iteration of the loop in Line~\ref{line:loop}, no new vertex that violates property $P$ is introduced, we observe that the last level $V_R$ is such that the entire neighborhood of every vertex in $V_R$ is stored in the $\ell_0$-samplers $N_{R,j}(u)$, for every $j$, which is a consequence of Lemma~\ref{lem:technical-1}.

    Thus, overall, we have established that each iteration of the main for loop fixes a vertex and does not introduce any new vertices that violate property $P$. Lemma~\ref{lem:key-1} thus establishes that Line~\ref{line:final-fix} turns $M$ into a globally maximal matching, which establishes that the output matching is maximal and completes the proof.

\end{proof}


\section{Space Lower Bound for \textsf{Maximal Matching}} \label{sec:lb}
In this section, we give our space lower bound for bounded-deletion streaming algorithms for \textsf{Maximal Matching}. To this end, in Subsection~\ref{sec:embedded-augmented-index} we define the \textsf{Embedded-Augmented-Index} problem, an extension of the \textsf{Augmented-Index} problem, and lower-bound its  communication complexity. Then, in Subsection~\ref{sec:reduction}, we show that a bounded-deletion streaming algorithm for \textsf{Maximal Matching} can be used to solve \textsf{Embedded-Augmented-Index}, which yields our main lower bound result.

\subsection{The \textsf{Embedded-Augmented-Index} Problem} \label{sec:embedded-augmented-index}
The $\textsf{Embedded-Augmented-Index}_{s,t}$ problem is a one-way two-party communication problem and is parameterized by two integers $s$ and $t$. Alice holds $s$ bitstrings $X := X_1, \dots, X_s \in \{0,1\}^{t}$. Bob holds two indices $I \in [s]$ and $J \in [t]$ as well as the suffix $X_I[J+1, \dots, t]$. Alice sends a single message to Bob, and, upon receipt, Bob needs to output the bit $X_I[J]$. Protocols can be randomized and need to be correct on every input with probability at least $\frac{1}{2} + \epsilon$, for some small constant $\epsilon > 0$.

Let $\mu$ denote the uniform input distribution, where all variables $(X, I, J)$ are chosen uniformly at random from their domains.

We now prove that the communication complexity of $\textsf{Embedded-Augmented-Index}_{s,t}$ is $\Omega(s \cdot t)$. 

\begin{theorem}\label{thm:embedded-aug-index}
    Every randomized communication protocol for $\textsf{Embedded-Augmented-Index}_{s,t}$ that is correct with probability $\frac{1}{2} + \epsilon$, for any $\epsilon > 0$, requires a message of size at least $$(1 - H_2(\frac{1}{2} - \epsilon)) \cdot s \cdot t \ .$$ 
\end{theorem}
\newcounter{counterEmbed}  
\setcounter{counterEmbed}{\value{theorem}}

\begin{proof}
    Let $\pi$ be a communication protocol for $\textsf{Embedded-Augmented-Index}_{s,t}$ that is correct with probability at least $\frac{1}{2} + \epsilon$, and let $\tilde{X}$ be the output generated by the protocol. Observe that $\tilde{X}$ is a function of Bob's input $I,J,X_I[J+1, \dots, t]$, Bob's private randomness $R_B$ and the public randomness $R$ used by the protocol, and the message $M$ sent from Alice to Bob. Then, by Fano's inequality (Lemma~\ref{lem:fano}), we obtain:
\begin{align*}
    H_{\mu}(X_I[J] \ | \ I,J,X_I[J+1, \dots, t], R_B, R, M) \le H_2(\frac{1}{2} - \epsilon) \ .     
\end{align*}
Since $H_{\mu}(X_I[J]) = 1$, we obtain:
\begin{align*}
    I_{\mu}(X_I[J] \ & : \ I,J,X_I[J+1, \dots, t], R_B, R, M)   \\
    & = H_{\mu}(X_I[J]) - H_{\mu}(X_I[J] \ | \ I,J,X_I[J+1, \dots, t], R_B, R, M) 
 \ge 1 - H_2(\frac{1}{2} - \epsilon) \ .     
\end{align*}
Let $\delta = 1 - H_2(\frac{1}{2} - \epsilon)$. Then, (please refer to the preliminaries for an overview of the information-theoretic arguments applied in the following derivation) 
\begin{align*}
  \delta & \le   I_{\mu}(X_I[J] \ : \ I,J,X_I[J+1, \dots, t], R_B, R, M) \\
  & = I_{\mu}(X_I[J] \ : \ I,J,X_I[J+1, \dots, t], R_B, R)  \\
  &  \hspace{1.5cm} + I_{\mu}(X_I[J] \ : \  M \ | \ I,J,X_I[J+1, \dots, t], R_B, R) & \text{Chain rule for MI} \\
  & = 0 + I_{\mu}(X_I[J] \ : \  M \ | \ I,J,X_I[J+1, \dots, t], R_B, R) & \mbox{RVs are independent} \\
  & = \frac{1}{t} \sum_{j = 1}^t I_{\mu}(X_I[j] \ : \  M \ | \ I,X_I[j+1, \dots, t], R_B, R, J=j) & \mbox{Def of conditional MI via exp} \\
  & = \frac{1}{t} \sum_{j = 1}^t I_{\mu}(X_I[j] \ : \  M \ | \ I,X_I[j+1, \dots, t], R_B, R) & \mbox{Independence of event $J=j$} \\
  & = \frac{1}{t} I_{\mu}(X_I \ : \  M \ | \ I, R_B, R) & \mbox{Chain rule for MI} \\
  & = \frac{1}{t \cdot s} \sum_{i=1}^s I_{\mu}(X_i \ : \  M \ | \ R_B, R, I = i) & \mbox{Def of conditional MI via exp} \\
  & \le \frac{1}{t \cdot s}  \sum_{i=1}^s I_{\mu}(X_i \ : \  M \ | \ R_B, R, X_1, \dots, X_{i-1}, I = i) & \mbox{Lemma~\ref{lem:property-mi}} \\
  & \le \frac{1}{t \cdot s}  \sum_{i=1}^s I_{\mu}(X_i \ : \  M \ | \ R_B, R, X_1, \dots, X_{i-1}) & \mbox{Independence of event $I=i$} \\
  & \le \frac{1}{t \cdot s} I_{\mu}(X \ : \  M \ | \ R_B, R) & \mbox{Chain rule for MI} \\ 
  & \le \frac{1}{t\cdot s} H_{\mu}(M \ | \ R_B, R)  & \mbox{Definition of MI} \\
  & \le \frac{1}{t \cdot s} H_{\mu}(M) \ . & \mbox{Conditioning reduces entropy}
\end{align*}
We thus obtain $H_{\mu}(M) \ge \delta \cdot t \cdot s$, which completes the proof.
\end{proof}

\subsection{From \textsf{Embedded-Augmented-Index} to \textsf{Maximal Matching}} \label{sec:reduction}

In this section, we will prove our lower bound result, showing that every bounded-deletion streaming algorithm for \textsf{Maximal Matching} requires space $\Omega(n \sqrt{K})$. The proof closely follows \cite{dk20}, adapted to \textsf{Embedded-Augmented-Index} as the underlying hard communication problem and somewhat simplified since we only require a lower bound for small constant factor approximations as they are provided by a \textsf{Maximal Matching} algorithm.

\setcounter{thmsaved}{\value{theorem}}
\setcounter{theorem}{\value{counterLB}}
\addtocounter{theorem}{-1}

\begin{theorem}
    Every bounded-deletion streaming algorithm that computes a \textsf{Maximal Matching} with probability at least $\frac{2}{3}$ requires space $\Omega( n \sqrt{K})$.
\end{theorem}
\setcounter{theorem}{\value{thmsaved}}
\addtocounter{theorem}{-1}
\begin{proof}
We will give a reduction to $\textsf{Embedded-Augmented-Index}_{s,t}$. For integers $s,t$ whose  values we will determine later, let $X = X_1, \dots, X_s \in \{0, 1\}^t, I \in [s]$ and $ J \in [t]$ be an $\textsf{Embedded-Augmented-Index}_{s,t}$ instance, and let $\mathcal{A}$ be a bounded-deletion streaming algorithm for \textsf{Maximal Matching} that outputs a maximal matching with probability $\frac{2}{3}$ on every instance. We assume that, if $\mathcal{A}$ fails, then it still outputs a matching, but it may not be maximal.

\vspace{0.1cm}
\noindent \textbf{Alice.} Given $X$, Alice constructs the bipartite graph $G=G_1 \ \dot{\cup} \ G_2 \  \dot{\cup} \ \dots \  \dot{\cup} \ G_s$, which is a disjoint union of the graphs $(G_i)_{1 \le i \le s}$. We now describe the construction of graph $G_i$, for any $ 1 \le i \le s$. The vertex set $A_i \cup B_i$ of $G_i$ is such that $|A_i| = |B_i| = 5 \cdot \sqrt{t}$. The edge set $E_i$ of $G_i$ is best described via its bipartite incidence matrix $M_i$. To obtain $M_i$, we first construct the matrix $\tilde{M}_i$, which has the same dimensions as $M_i$. The construction process is as follows:
\begin{enumerate}
    \item The top-left square sub-matrix of $\tilde{M}_i$ with side length $\sqrt{t}$ is filled with the entries of $X_i$ from left-to-right and top-to-bottom.
    \item All other entries of $\tilde{M}_i$ are filled with uniform random bits sampled from public randomness.
    \item Alice and Bob sample a random binary square matrix $Y_i$ with side length $5 \sqrt{t}$ such that each entry is $1$ with probability $\frac{1}{2}$. They compute the entry-wise XOR between $\tilde{M}_i$ and $Y_i$. Let $\tilde{M}_i'$ denote the resulting matrix.
    \item Alice and Bob sample random permutations $\sigma_i, \pi_i: [5 \sqrt{t}] \rightarrow [5 \sqrt{t}]$ from public randomness. The matrix $M_i$ is obtained from $\tilde{M}_i'$ by permuting the rows and columns with the permutations $\sigma_i$ and $\pi_i$, respectively.
\end{enumerate}
Alice then runs algorithm $\mathcal{A}$ on the edges of graph $G$, presented to the algorithm in random order.

\vspace{0.1cm}
\noindent \textbf{Bob.} Recall that Bob holds the indices $I \in [s]$ and $J \in [t]$ as well as the suffix $X_I[J+1, \dots, t]$.
Bob introduces edge deletions, but only into the graph $G_I$. To this end, consider the matrix $\tilde{M}_I$, and let $a,b$ be such that the entry $\tilde{M}_I[a,b]$ corresponds to the bit $X_I[J]$. Then, observe that Bob knows all entries of the submatrix $\tilde{M}_I(a,b)$ with top-left corner at position $(a,b)$ and bottom-right corner being the bottom-right corner of $\tilde{M}_I$, except the entry $\tilde{M}_I(a,b)[a,b]$, either because the bits $X_I[J+1, t]$ correspond to these entries or the entries were constructed from public randomness. Denote by 
$$\mathcal{I} = \{(i,j) \ | \ a \le i \le 5 \sqrt{t}, b \le j \le 5 \sqrt{t} \mbox{ and } i - a \neq j - b\} \, $$ 
i.e., the indices in $\mathcal{I}$ describe the entries of $\tilde{M}_I$ that constitute off-diagonal entries in the submatrix $\tilde{M}_I(a,b)$. Then, for each index $(i,j) \in \mathcal{I}$, if $M_{\sigma(i), \pi(j)} = 1$ then Bob feeds an edge deletion that turns this entry to $0$ into the algorithm $\mathcal{A}$, i.e., $M_{\sigma(i), \pi(j)} = 0$ after the deletion. All deletions are fed into $\mathcal{A}$ in random order.

Bob then obtains the output matching $OUT$ produced by algorithm $\mathcal{A}$.

\vspace{0.1cm}
\noindent \textbf{Analysis.} 
 We assume from now on that the algorithm $\mathcal{A}$  does not fail. Then, we  claim that, with high probability, $|OUT| \ge \frac{1}{4} \cdot 5 \sqrt{t} - o(\sqrt{t})$. Indeed, observe that none of the diagonal entries in $\tilde{M}_i'$ are subsequently deleted, and the edges corresponding to the $1$-entries of this diagonal form a matching. By concentration bounds, with high probability, there are at least $\frac{1}{2} \cdot 5 \sqrt{t} - o(\sqrt{t})$ $1$-entries. The claim follows from the observation that $M_i$ is obtained from $\tilde{M}_i'$ by permuting the rows and columns, which are operations that preserve the matching size, and by the fact that a maximal matching is at least half the size of a maximum matching.

Next, we say that the entry $X_I[J]$ materializes as an edge in graph $G$ if $\tilde{M}_I'[a,b] = 1$. Recall that $\tilde{M}_I'[a,b] = \tilde{M}_I[a,b] \ XOR \ Y[a,b]$, where $Y[a,b]$ is a uniform random bit. Hence, $X_I[J]$ materializes as an edge with probability $1/2$.

We will now argue that, if the edge corresponding to $X_I[J]$ materializes, then it is also contained in $OUT$ with constant probability. Indeed, due to the application of the random permutations $\sigma_i$ and $\pi_i$, each of the diagonal entries of $\tilde{M}_I'(a,b)$ with value $1$ are reported in the matching $OUT$ with equal probability. Hence, conditioned on the entry $X_I[J]$ materializing as an edge in $G_I$, this edge is contained in $OUT$ with probability at least 
$$\frac{|OUT| - \sqrt{t}}{5 \sqrt{t}} \ge \frac{\frac{5}{4} \sqrt{t} - o(\sqrt{t}) - \sqrt{t}}{5 \sqrt{t} } \ge \frac{1}{20} - o(1) \ , $$ since $OUT$ contains at least $|OUT| - \sqrt{t}$ diagonal entries from the matrix $\tilde{M}_I'(a,b)$. 

Hence, the probability that the edge corresponding to $X_I[J]$ materializes and is reported in $OUT$ is at least $\frac{1}{40} - o(1)$. If the edge is not observed in $OUT$ then the algorithm reports either $0$  or $1$ as a guess for $X_I[J]$, each with probability $\frac{1}{2}$. 

Using the following definition of $p$,
\begin{align*}
    p & = \Pr[\mathcal{A} \mbox{ succeeds and } X_I[J] \mbox{ materializes as an edge and is reported in } OUT] \\
    & \ge \frac{2}{3} \cdot (\frac{1}{40}-o(1)) = \frac{1}{60} - o(1) \ ,    
\end{align*} 
we can bound the overall success probability of this strategy as follows:
\begin{align*}
  \Pr[X_I[J]  \mbox{ recovered correctly}] & 
  \ge p + (1-p) \cdot \frac{1}{2} \\
  & = \frac{1}{60} + \frac{59}{60} \cdot \frac{1}{2} - o(1)  = \frac{1}{2} + \frac{1}{120} - o(1) \ .
\end{align*}


By Theorem~\ref{thm:embedded-aug-index}, the algorithm therefore requires space $\Omega(s \cdot t)$. Lastly, to ensure that at most $K$ deletions are introduced, we set $t = K$, and to ensure that the input graph has $n$ vertices, we set $s = \frac{n}{2 \cdot 5 \sqrt{t}}$. The lower bound thus gives $\Omega(s \cdot t) = \Omega(\frac{n}{\sqrt{K}} \cdot K) = \Omega(n \cdot \sqrt{K})$, as desired.
\end{proof}

\section{Deterministic Algorithm and Space Lower Bound for \textsf{Maximal Matching}} \label{sec:det}

The aim of this section is to settle the deterministic space complexity of \textsf{Maximal Matching} in bounded-deletion streams. We will prove the following theorem:

\setcounter{thmsaved}{\value{theorem}}
\setcounter{theorem}{\value{counterDET}}
\addtocounter{theorem}{-1}
\begin{theorem}
There is a single-pass streaming algorithm that uses $\tilde{O}(n \cdot K)$ space and outputs a maximal matching in any dynamic graph stream with at most $K$ deletions. Moreover, any deterministic algorithm for \textsf{Maximal Matching} requires $\Omega(n \cdot K)$ space.
\end{theorem}
\setcounter{theorem}{\value{thmsaved}}

A deterministic $O(n \cdot K \cdot \log n)$ algorithm that matches the result stated in the previous theorem will be established in Subsection~\ref{sec:det-alg}, and our our $\Omega(n \cdot K)$ space lower bound result will be given in Subsection~\ref{sec:det-lb}.

\subsection{Algorithm} \label{sec:det-alg}
Our deterministic $O(n \cdot K \cdot \log n)$ space algorithm can easily be obtained from the tools developed in Section~\ref{sec:alg}. Our algorithm computes a hierachical matching on $L=K+1$ levels via Algorithm~\ref{alg:hierachical-matching} on the substream of edge insertions, which uses space $\tilde{O}(n \cdot K)$. At the same time, the algorithm stores the (multi-)set of edge deletions $D$, using space $O(K \cdot \log n)$. 

We now distinguish two cases. First, suppose that the hierachical matching data structure consists of $L$ non-empty levels. Then, we pick a level, say $M_i$, that was not subject to any deletions, which must exist due to the pigeonhole principle, and greedily add all edges from the matchings $M_1, \dots, M_{i-1} \setminus D$ that are not deleted to $M_i$. Then, by Lemma~\ref{lem:key-1}, the resulting matching is maximal.

In the second case, suppose that the matching data structure contains fewer than $L$ non-empty levels. Then, the algorithm has stored all edges of the input graph and can even produce a maximum matching. 

This establishes the following theorem:

\begin{theorem} \label{thm:det-alg}
    There is a deterministic insertion-deletion streaming algorithms for \textsf{Maximal Matching} with space $O(n \cdot K \cdot \log n)$, where $K$ is an upper bound on the number of deletions in the stream. 
\end{theorem}

\subsection{Lower Bound} \label{sec:det-lb}
We now give our $\Omega(n \cdot K)$ space lower bound for deterministic insertion-deletion streaming algorithms for \textsf{Maximal Matching}, where $K$ is an upper bound on the number of deletions in the input stream. 

We prove our lower bound in the one-way two-party communication setting, where Alice holds edge insertions $E$,  Bob holds edge deletions $D \subseteq E$ with $|D| \le K$, Alice sends a message to Bob, and Bob output the result of the computation. It is well-known that a lower bound on the size of the message required to solve the problem also constitutes a lower bound on the space required by one-pass streaming algorithms. The following theorem thus implies the lower bound stated in Theorem~\ref{thm:det_maximal}.

\begin{theorem}
 Every deterministic one-way two-party protocol for \textsf{Maximal Matching}, where Bob holds at most $K$ edge deletions, requires a message of size at least $\frac{n \cdot K}{8}$ bits.
\end{theorem}

\begin{proof}
 Let $\mathcal{H}_K$ be the family of bipartite graphs $G=(A, B, E)$ with $|A| = K$, $|B| = 3 \cdot K$, and the degree of every $A$-vertex is $2 \cdot K$.
 Furthermore, for parameters $K$ and $n$ with $n$ being a multiple of $4 \cdot K$, let $\mathcal{G}_K(n)$ be the family of graphs obtained as the disjoint union of $n / (4K)$ graphs picked from $\mathcal{H}$.
 Then, the number of graphs in $\mathcal{G}_K(n)$ is:
 \begin{align*}
  |\mathcal{G}_K(n)| = {3 \cdot K \choose 2 \cdot K}^{K \cdot \frac{n}{4 K}} 
   = {3 \cdot K \choose 2 \cdot K}^{\frac{n}{4}}  \ , 
 \end{align*}
and any encoding of this set of graphs requires at least
\begin{align}
 \log \left( {3 \cdot K \choose 2 \cdot K}^{\frac{n}{4}} \right) = \frac{n}{4} \log {3 \cdot K \choose 2 \cdot K} \text{ bits.} \label{eqn:184}
\end{align}

 Next, let $G' \in \mathcal{G}_K(n)$ be any graph, and let $G$ be obtained from $G'$ by removing at most $K$ arbitrary edges. Then, observe that every maximal matching in $G$ matches all $A$-vertices. To see this, suppose for the sake of a contradiction that this is not the case and there exists an $a \in A$ and a maximal matching $M$ in $G$ such that $a$ is not matched in $M$. Observe that $\deg_G(a) \ge K$ holds, since at most $K$ edges are deleted and the degree of $a$ in $G'$ is $2K$. For $M$ to be maximal, since $a$ is not matched in $M$, all of $a$'s neighbors must be matched in $M$, which implies that the component that contains $a$ must contribute to $M$ with at least $\deg_G(a) \ge K$ edges. Recall, however, that the component that contains $a$ consists of only $K$ $A$-vertices. Then, since the component is bipartite and $a$ is not matched, at most $K-1$ edges of the component can contribute to $M$, a contradiction. Every $A$-vertex must therefore be matched in any maximal matching in $G$.
  
 Let $\mathcal{P}$ be a communication protocol as in the statement of the theorem, and suppose that the message ${\sf Msg}$ sent from Alice to Bob is of length at most $|{\sf Msg}|$. 
We will now argue that $\mathcal{P}$ yields an efficient encoding of the graphs in $\mathcal{G}_K(n)$ using a number of bits that depends on $|{\sf Msg}|$. Since this encoding must require at least the number of bits stated in Equation~\ref{eqn:184}, we obtain a lower bound on $|Msg|$.

Suppose that Alice's input is a graph $G \in \mathcal{G}_K(n)$. We will now argue that Bob can learn $K+1$ incident edges to every $A$-vertex in $G$ from the message ${\sf Msg}$ sent from Alice to Bob, as follows:

Let $a \in A$ be any vertex. Given ${\sf Msg}$, first Bob does not delete any edges and completes the protocol $\mathcal{P}$ to produce a maximal matching $M_1$. As previously argued, $M_1$ must match every $A$-vertex, and denote by $e_1$ the edge in $M_1$ that is incident to $a$. Again, starting with the message ${\sf Msg}$, Bob simulates another execution of $\mathcal{P}$ and feeds the deletion $e_1$ into $\mathcal{P}$. The protocol then produces another maximal matching $M_2$, and let $e_2 \in M_2$ denote the edge incident to $a$. We observe that $e_2$ is a different edge to $e_1$ since the edge $e_1$ was explicitly deleted. Repeating this process, Bob again simulates the execution of $\mathcal{P}$ and feeds both $e_1$ and $e_2$ as edge deletions into $\mathcal{P}$, which allows him to learn yet another edge $e_3$ that is different to $e_1$ and $e_2$ incident to $a$. After repeating this process $K$ times, Bob has learned $K+1$ edges incident to $a$. Last, observe that Bob can execute this strategy for every $A$-vertex, allowing Bob to learn $K+1$ edges incident on every $A$-vertex.

Since the protocol is deterministic, Alice does in fact already know the set of $\frac{n}{4} \cdot (K+1)$ edges that Bob learns from the message ${\sf Msg}$. Using ${\sf Msg}$, Alice can encode any input graph in $\mathcal{G}_K(n)$ using the following encoding: 

First, recall that ${\sf Msg}$ encodes $K+1$ edges for every $A$-vertex. It remains to encode the remaining $2K-(K+1) = K-1$ edges for each $A$-vertex. Observe that, for each $A$-vertex, there are only $3K - (K+1) = 2K-1$ $B$-vertices available, i.e., the remaining neighbors in their respective components. Hence, the remaining neighborhood of each $A$-vertex can be encoded with $\log{2K-1 \choose K-1}$ bits. 

This encoding uses
\begin{align}
  \frac{n}{4} \cdot \log {2K-1 \choose K-1} + |{\sf Msg}|  \label{eqn:431}
\end{align}
bits in total, and combined with Equation~\ref{eqn:184}, we obtain:
\begin{align}
 \frac{n}{4} \cdot \log {2K-1 \choose K-1} + |{\sf Msg}| & \ge \frac{n}{4} \log {3 \cdot K \choose 2 \cdot K} \ , \mbox{ or} \nonumber \\
 |{\sf Msg}| & \ge \frac{n}{4} \log {3 \cdot K \choose 2 \cdot K} - \frac{n}{4} \cdot \log {2K-1 \choose K-1} \ . \label{eqn:571}
\end{align}
Next, since $\frac{a+1}{b+1} \cdot {a \choose b} = {a+1 \choose b+1}$ holds, which further implies ${a+(K+1) \choose b+(K+1)} \ge {a \choose b} \cdot \left( \frac{a+(K+1)}{b+(K+1)} \right)^{K+1}$, we obtain 
$${3 \cdot K \choose 2 \cdot K} \ge {2K - 1 \choose K-1} \cdot \left(\frac{3}{2}\right)^{K+1} \ , $$
and using this in Inequality~\ref{eqn:571} yields
\begin{align*}
 |{\sf Msg}| & \ge \frac{n}{4} \left( \log \left( {2K-1 \choose K-1} \left(\frac{3}{2} \right)^{K+1} \right)  -  \log  {2K-1 \choose K-1}  \right) \\
 & = \frac{n}{4} \left( \log  {2K-1 \choose K-1} + (K+1) \log(\frac{3}{2})   -  \log  {2K-1 \choose K-1}  \right) \\
 & = \frac{n}{4} \left(  (K+1) \log(\frac{3}{2})    \right) \ge \frac{n \cdot K}{8} ,
\end{align*}
which completes the proof.
\end{proof}

\section{$O(1)$-Approximation to \textsf{Maximum Matching} using $\tilde{O}(n + K)$ Space} \label{sec:maximum-matching}

In this section, we show that, for any $\epsilon > 0$, a $(2+\epsilon)$-approximation to \textsf{Maximum Matching} can be computed using $\tilde{O}(n+K/\epsilon)$ space, thus showing that there is a phase transition in space complexity as one moves from requiring maximality to simply a near-optimal solution in terms of size.

At the heart of our algorithm lies a variant of the \textsf{Hierachical-Greedy} algorithm (listed in Algorithm~\ref{alg:hierachical-matching}) that we denote by \textsf{Budgeted-Hierachical-Greedy} (listed in Algorithm~\ref{alg:budgeted-hierachical-matching}). Similar to \textsf{Hierachical-Greedy}, this algorithm is also executed on the substream of edge insertions. While \textsf{Hierachical-Greedy($L$)} is parametrized by the number of desired matchings in the computed matching hierarchy $L$, \textsf{Budgeted-Hierachical-Greedy($B$)} computes a hierachical matching structure that consists of  (at most) $B$ edges.

\begin{algorithm}
    \begin{algorithmic}[1]
        \REQUIRE Input stream of edge insertions (multi-edges allowed), integer parameter $B \ge 1$
        \STATE $M_1 \gets \varnothing$, $L \gets 1$
        \WHILE{stream not empty}
            \STATE Let $e$ be the next edge in the stream
            \STATE Let $\ell \in [L]$ be the smallest integer such that $M_{\ell} \cup \{e\}$ is a matching, if no such matching exists then let $L \gets L+1$, $M_L = \varnothing$, and let $\ell \gets L$
            \STATE $M_{\ell} \gets M_{\ell} \cup \{e\}$
            \IF{$\sum_{i=1}^{L} |M_i| = B+1$} \label{line:581}
                \STATE Remove an arbitrary edge from matching $M_L$
                \IF {$M_L = \varnothing$}
                    \STATE Delete $M_L$ and let $L \gets L - 1$
                \ENDIF
            \ENDIF
        \ENDWHILE
        \RETURN $(M_1, M_2, \dots, M_{L})$
    \end{algorithmic}
    \caption{\textsf{Budgeted-Hierachical-Greedy($B$)}  \label{alg:budgeted-hierachical-matching}}
\end{algorithm}

Similar to \textsf{Hierachical-Greedy}, for each incoming edge, \textsf{Budgeted-Hierachical-Greedy} finds the smallest level such that the edge can be inserted into this level. This may require creating a new level if the current edge cannot be inserted into any of the existing levels. In Line~\ref{line:581}, the algorithm checks that the total number of edges stored does not exceed the budget parameter $B$. If it does exceed the budget then the algorithm removes an arbitrary edge from the highest level. This makes sure that the structure consists of at most $B$ edges, and, if the edge deleted was the only edge in the matching in the highest level, then the matching in the highest level is deleted and the number of levels is decreased by one.

We will now show that, when using a budget $B = n + \frac{K}{\epsilon}$, then the matchings returned by the run \textsf{Budgeted-Hierachical-Greedy($B$)} contain a $(2+\epsilon)$-approximation to $\textsf{Maximum Matching}$, even after the substream of deletions is applied to the returned matchings.

To this end, we first give a technical lemma that establishes key  properties of \textsf{Budgeted-Hierachical-Greedy}:

\begin{lemma} 
\label{lem:prefixHM}
Let $M_1, \dots, M_L$ be the output of \textsf{Budgeted-Hierachical-Greedy($B$)}. Then the following two properties hold:

\begin{itemize}
\item[($P_1$)]
The algorithm never discarded an edge from the matchings $M_1, M_2, \dots, M_{L-1}$. Hence, they constitute a $(L - 1)$-level hierarchical matching data structure, or, equivalently, $M_1, \dots, M_{L-1}$ also constitutes the output of \textsf{Hierachical-Greedy}($L-1$).

\item[($P_2$)]
The algorithm either never discarded an edge or the matchings $M_1, M_2, \dots, M_{L}$ together contain $B$ edges.
\end{itemize}
\end{lemma}

\begin{proof}
Regarding $P_1$, suppose for the sake of a contradiction that a level-$i$ edge with $i \in [L-1]$ was deleted in some iteration $J$ of the algorithm. This implies that, in iteration $J$, the algorithm maintained only $i$ matchings in total since deletions take place only at the highest level matching. We now claim that it is impossible for the algorithm to have created any additional levels beyond level $i$ in any subsequent iteration, after $J$. Since the budget of $B$ was already reached in iteration $J$, from this moment onwards, whenever an incoming edge is inserted into the matching data structure, an edge of the current highest level matching is deleted. If an incoming edge temporarily leads to the creation of a new level $i+1$ then the algorithm would have immediately deleted this edge again, keeping the number of levels to $i$. This is a contradiction to the fact that the algorithm created $L \ge i+1$ levels, which establishes $P_1$.

To prove property $P_2$, we observe that deletions are only triggered once the number of edges in the matchings stored reaches $B$. Once this happens, then whenever another edge is inserted into the matchings, some edge is also deleted, making sure that the total number of edges in the matchings remains $B$ until the end of the algorithm. This establishes $P_2$.
\end{proof}

We are now ready to specify our complete algorithm for achieving a $(2+\epsilon)$-approximation. We set $B = n+K/\epsilon$ and execute \textsf{Budgeted-Hierachical-Greedy}($B$). At the same time, we also separately store the set of $K$ deletions encountered in the stream. Let $M_1, M_2, \dots, M_{L}$ be the output of the run of \textsf{Budgeted-Hierachical-Greedy}($B$). We then apply the stored deletions to the matchings $M_1, M_2, \dots, M_{L}$ and output a largest matching among the edges $\cup_{1 \le i \le L} M_i$.

\setcounter{thmsaved}{\value{theorem}}
\setcounter{theorem}{\value{counterUB2}}
\addtocounter{theorem}{-1}

\begin{theorem}
For any $\epsilon > 0$, there is a deterministic single-pass streaming algorithm that uses $O((n + K/\epsilon) \cdot \log n)$ space and outputs a $(2+\epsilon)$-approximate matching in any dynamic graph stream with at most $K$ deletions.
\end{theorem}

\setcounter{theorem}{\value{thmsaved}}

\begin{proof}
First, regarding space, the algorithm stores overall  $O(n + \frac{K}{\epsilon})$ edges. Accounting $O(\log n)$ bits for the space required to store an edge, we obtain the claimed space bound.

Regarding the approximation factor, recall that $B = n + \frac{K}{\epsilon}$, denote by $M_1, \dots, M_L$ the output of \textsf{Budgeted-Hierachical-Greedy}($B$), and denote by $\tilde{M}_1, \dots, \tilde{M}_L$ the matchings $M_1, \dots, M_L$ after the deletions were applied. We distinguish two cases:

In the first case, we suppose that the matchings $M_1, \dots, M_{L}$ contain fewer than $B$ edges. Then, by Property P2 proved in Lemma~\ref{lem:prefixHM}, $\cup_{i=1}^{L}M_i$ consitutes the {\em entire stream} of inserted edges, and $\cup_{i=1}^L \tilde{M}_i$  constitutes all edges of the graph described by the input stream. A maximum matching in $\cup_{i=1}^L \tilde{M}_i$ is therefore also a maximum matching in the input graph. 

In the second case, suppose that $M_1, \dots, M_{L}$ contains $B$ edges.  Then, by property $P_1$ proved in Lemma~\ref{lem:prefixHM}, the matchings $M_1, , \dots, M_{L-1}$ constitute a hierarchical matching data structure with $L-1$ levels. Observe that $|\cup_{i=1}^{L-1}M_i|$ must be at least 
$$|\cup_{i=1}^{L-1}M_i| = B - |M_L| \ge B - \frac{n}{2} = (n/2) + K/\epsilon > n/2 \ ,$$ implying that $L-1 \ge 1$, that is,  we have at least one {\em complete} maximal matching. Let $m_i$ denote the number of edges in the matching $M_i$ for $1 \le i \le L-1$. 
We claim that there exists an index $j \in [L-1]$ such that $\tilde{M}_j \ge (1-\epsilon) M_j$, i.e., at most an $\epsilon$-fraction of $M_j$'s edges are deleted. This is easily seen using a proof by contradiction: if not, then the total number of deletions must be at least
$$\sum_{i=1}^{L-1} \epsilon \cdot m_i \ge \epsilon \cdot ((n/2) + K/\epsilon) > K,$$
which is a contradiction to our assumption that there are at most $K$ deletions. Consider now the following thought experiment. We do not apply any of the deletions to the matching $M_j$. We now perform the downward extension of $M_j$, adding to it greedily the edges in $\tilde{M}_1, \dots, \tilde{M}_{j-1}$. By Lemma~\ref{lem:key-1}, the resulting matching, say $M$, is a maximal matching in the final graph $G$. We can now remove at most $\epsilon \cdot m_j \le \epsilon |M|$ edges that were deleted in $M_j$, obtaining a $(2+\epsilon)$-approximate matching, as desired.
\end{proof}

\section{Conclusion}\label{sec:conclusion}
In this work, we initiated the study of \textsf{Maximal Matching} in  bounded-deletion streams. We settled the space complexity of both randomized and deterministic algorithms up to $\poly \log$ factors by given algorithms as well as lower bounds, where the randomized space complexity is $\tilde{\Theta}(n \cdot \sqrt{K})$, and the deterministic space complexity is $\tilde{\Theta}(n \cdot K)$. Our results constitute the first trade-off results between space complexity and the number of deletions for a fundamental graph problem. 

We hope that our work will spark work on other fundamental graph problems in the bounded-deletion stream setting, especially when the presence of deletions seems to significantly increase the space complexity of the streaming algorithm. Two natural and closely related problems for future investigation are approximate \textsf{Maximum Matching} and approximate \textsf{Minimum Vertex Cover}, where a refined understanding of space complexity in terms of the deletion parameter $K$ will reveal how the space complexity changes as we gradually move from insertion-only to unrestricted deletions streams. Another problem that will be interesting to study through the lens of bounded deletions is the problem of building low stretch graph spanners. Similar to \textsf{Maximum Matching} and \textsf{Minimum Vertex Cover}, there is a polynomial-factor separation in the space needed to build a low stretch spanner in an insertion-only stream, and in dynamic streams with unrestricted deletions.


Finally, we believe that our algorithmic approach of building a complementary pair of sketches, one that is  non-linear, focusing on the insertions, and another one that is non-uniform and linear, focusing on the dynamic component of the stream, will also prove useful for other graph problems.


\bibliography{kk25}

\begin{thebibliography}{10}

\bibitem{agm12}
Kook~Jin Ahn, Sudipto Guha, and Andrew McGregor.
\newblock Analyzing graph structure via linear measurements.
\newblock In Yuval Rabani, editor, {\em Proceedings of the Twenty-Third Annual
  {ACM-SIAM} Symposium on Discrete Algorithms, {SODA} 2012, Kyoto, Japan,
  January 17-19, 2012}, pages 459--467. {SIAM}, 2012.

\bibitem{ahlw16}
Yuqing Ai, Wei Hu, Yi~Li, and David~P. Woodruff.
\newblock New characterizations in turnstile streams with applications.
\newblock In Ran Raz, editor, {\em 31st Conference on Computational Complexity,
  {CCC} 2016, May 29 to June 1, 2016, Tokyo, Japan}, volume~50 of {\em LIPIcs},
  pages 20:1--20:22. Schloss Dagstuhl - Leibniz-Zentrum f{\"{u}}r Informatik,
  2016.

\bibitem{ams96}
Noga Alon, Yossi Matias, and Mario Szegedy.
\newblock The space complexity of approximating the frequency moments.
\newblock In Gary~L. Miller, editor, {\em Proceedings of the Twenty-Eighth
  Annual {ACM} Symposium on the Theory of Computing, Philadelphia,
  Pennsylvania, USA, May 22-24, 1996}, pages 20--29. {ACM}, 1996.

\bibitem{akl16}
Sepehr Assadi, Sanjeev Khanna, Yang Li, and Grigory Yaroslavtsev.
\newblock Maximum matchings in dynamic graph streams and the simultaneous
  communication model.
\newblock In Robert Krauthgamer, editor, {\em Proceedings of the Twenty-Seventh
  Annual {ACM-SIAM} Symposium on Discrete Algorithms, {SODA} 2016, Arlington,
  VA, USA, January 10-12, 2016}, pages 1345--1364. {SIAM}, 2016.

\bibitem{alt21}
Sepehr Assadi, S.~Cliff Liu, and Robert~E. Tarjan.
\newblock An auction algorithm for bipartite matching in streaming and
  massively parallel computation models.
\newblock In Hung~Viet Le and Valerie King, editors, {\em 4th Symposium on
  Simplicity in Algorithms, {SOSA} 2021, Virtual Conference, January 11-12,
  2021}, pages 165--171. {SIAM}, 2021.

\bibitem{as22}
Sepehr Assadi and Vihan Shah.
\newblock An asymptotically optimal algorithm for maximum matching in dynamic
  streams.
\newblock In Mark Braverman, editor, {\em 13th Innovations in Theoretical
  Computer Science Conference, {ITCS} 2022, January 31 - February 3, 2022,
  Berkeley, CA, {USA}}, volume 215 of {\em LIPIcs}, pages 9:1--9:23. Schloss
  Dagstuhl - Leibniz-Zentrum f{\"{u}}r Informatik, 2022.

\bibitem{ejwy22}
Omri Ben{-}Eliezer, Rajesh Jayaram, David~P. Woodruff, and Eylon Yogev.
\newblock A framework for adversarially robust streaming algorithms.
\newblock {\em J. {ACM}}, 69(2):17:1--17:33, 2022.

\bibitem{kk20}
Lidiya~Khalidah binti Khalil and Christian Konrad.
\newblock Constructing large matchings via query access to a maximal matching
  oracle.
\newblock In Nitin Saxena and Sunil Simon, editors, {\em 40th {IARCS} Annual
  Conference on Foundations of Software Technology and Theoretical Computer
  Science, {FSTTCS} 2020, December 14-18, 2020, {BITS} Pilani, {K} {K} Birla
  Goa Campus, Goa, India (Virtual Conference)}, volume 182 of {\em LIPIcs},
  pages 26:1--26:15. Schloss Dagstuhl - Leibniz-Zentrum f{\"{u}}r Informatik,
  2020.

\bibitem{bgw20}
Mark Braverman, Sumegha Garg, and David~P. Woodruff.
\newblock The coin problem with applications to data streams.
\newblock In Sandy Irani, editor, {\em 61st {IEEE} Annual Symposium on
  Foundations of Computer Science, {FOCS} 2020, Durham, NC, USA, November
  16-19, 2020}, pages 318--329. {IEEE}, 2020.

\bibitem{ccehmmv16}
Rajesh Chitnis, Graham Cormode, Hossein Esfandiari, MohammadTaghi Hajiaghayi,
  Andrew McGregor, Morteza Monemizadeh, and Sofya Vorotnikova.
\newblock Kernelization via sampling with applications to finding matchings and
  related problems in dynamic graph streams.
\newblock In Robert Krauthgamer, editor, {\em Proceedings of the Twenty-Seventh
  Annual {ACM-SIAM} Symposium on Discrete Algorithms, {SODA} 2016, Arlington,
  VA, USA, January 10-12, 2016}, pages 1326--1344. {SIAM}, 2016.

\bibitem{cgsv24}
Chi{-}Ning Chou, Alexander Golovnev, Madhu Sudan, and Santhoshini Velusamy.
\newblock Sketching approximability of all finite csps.
\newblock {\em J. {ACM}}, 71(2):15:1--15:74, 2024.

\bibitem{cf14}
Graham Cormode and Donatella Firmani.
\newblock A unifying framework for {$(\ell_0)$}-sampling algorithms.
\newblock {\em Distributed Parallel Databases}, 32(3):315--335, 2014.

\bibitem{cjmm17}
Graham Cormode, Hossein Jowhari, Morteza Monemizadeh, and S.~Muthukrishnan.
\newblock The sparse awakens: Streaming algorithms for matching size estimation
  in sparse graphs.
\newblock In Kirk Pruhs and Christian Sohler, editors, {\em 25th Annual
  European Symposium on Algorithms, {ESA} 2017, September 4-6, 2017, Vienna,
  Austria}, volume~87 of {\em LIPIcs}, pages 29:1--29:15. Schloss Dagstuhl -
  Leibniz-Zentrum f{\"{u}}r Informatik, 2017.

\bibitem{ct06}
Thomas~M. Cover and Joy~A. Thomas.
\newblock {\em Elements of information theory {(2.} ed.)}.
\newblock Wiley, 2006.

\bibitem{dk20}
Jacques Dark and Christian Konrad.
\newblock Optimal lower bounds for matching and vertex cover in dynamic graph
  streams.
\newblock In Shubhangi Saraf, editor, {\em 35th Computational Complexity
  Conference, {CCC} 2020, July 28-31, 2020, Saarbr{\"{u}}cken, Germany (Virtual
  Conference)}, volume 169 of {\em LIPIcs}, pages 30:1--30:14. Schloss Dagstuhl
  - Leibniz-Zentrum f{\"{u}}r Informatik, 2020.

\bibitem{fkmsz04}
Joan Feigenbaum, Sampath Kannan, Andrew McGregor, Siddharth Suri, and Jian
  Zhang.
\newblock On graph problems in a semi-streaming model.
\newblock In Josep D{\'{\i}}az, Juhani Karhum{\"{a}}ki, Arto Lepist{\"{o}}, and
  Donald Sannella, editors, {\em Automata, Languages and Programming: 31st
  International Colloquium, {ICALP} 2004, Turku, Finland, July 12-16, 2004.
  Proceedings}, volume 3142 of {\em Lecture Notes in Computer Science}, pages
  531--543. Springer, 2004.

\bibitem{hssw12}
Magn{\'{u}}s~M. Halld{\'{o}}rsson, Xiaoming Sun, Mario Szegedy, and Chengu
  Wang.
\newblock Streaming and communication complexity of clique approximation.
\newblock In Artur Czumaj, Kurt Mehlhorn, Andrew~M. Pitts, and Roger
  Wattenhofer, editors, {\em Automata, Languages, and Programming - 39th
  International Colloquium, {ICALP} 2012, Warwick, UK, July 9-13, 2012,
  Proceedings, Part {I}}, volume 7391 of {\em Lecture Notes in Computer
  Science}, pages 449--460. Springer, 2012.

\bibitem{hkmms22}
Avinatan Hassidim, Haim Kaplan, Yishay Mansour, Yossi Matias, and Uri Stemmer.
\newblock Adversarially robust streaming algorithms via differential privacy.
\newblock {\em J. {ACM}}, 69(6):42:1--42:14, 2022.

\bibitem{hrr98}
Monika~Rauch Henzinger, Prabhakar Raghavan, and Sridhar Rajagopalan.
\newblock Computing on data streams.
\newblock In James~M. Abello and Jeffrey~Scott Vitter, editors, {\em External
  Memory Algorithms, Proceedings of a {DIMACS} Workshop, New Brunswick, New
  Jersey, USA, May 20-22, 1998}, volume~50 of {\em {DIMACS} Series in Discrete
  Mathematics and Theoretical Computer Science}, pages 107--118. {DIMACS/AMS},
  1998.

\bibitem{jst11}
Hossein Jowhari, Mert Saglam, and G{\'{a}}bor Tardos.
\newblock Tight bounds for lp samplers, finding duplicates in streams, and
  related problems.
\newblock In Maurizio Lenzerini and Thomas Schwentick, editors, {\em
  Proceedings of the 30th {ACM} {SIGMOD-SIGACT-SIGART} Symposium on Principles
  of Database Systems, {PODS} 2011, June 12-16, 2011, Athens, Greece}, pages
  49--58. {ACM}, 2011.

\bibitem{kp20}
John Kallaugher and Eric Price.
\newblock Separations and equivalences between turnstile streaming and linear
  sketching.
\newblock In Konstantin Makarychev, Yury Makarychev, Madhur Tulsiani, Gautam
  Kamath, and Julia Chuzhoy, editors, {\em Proceedings of the 52nd Annual {ACM}
  {SIGACT} Symposium on Theory of Computing, {STOC} 2020, Chicago, IL, USA,
  June 22-26, 2020}, pages 1223--1236. {ACM}, 2020.

\bibitem{k15}
Christian Konrad.
\newblock Maximum matching in turnstile streams.
\newblock In Nikhil Bansal and Irene Finocchi, editors, {\em Algorithms - {ESA}
  2015 - 23rd Annual European Symposium, Patras, Greece, September 14-16, 2015,
  Proceedings}, volume 9294 of {\em Lecture Notes in Computer Science}, pages
  840--852. Springer, 2015.

\bibitem{k18}
Christian Konrad.
\newblock A simple augmentation method for matchings with applications to
  streaming algorithms.
\newblock In Igor Potapov, Paul~G. Spirakis, and James Worrell, editors, {\em
  43rd International Symposium on Mathematical Foundations of Computer Science,
  {MFCS} 2018, August 27-31, 2018, Liverpool, {UK}}, volume 117 of {\em
  LIPIcs}, pages 74:1--74:16. Schloss Dagstuhl - Leibniz-Zentrum f{\"{u}}r
  Informatik, 2018.

\bibitem{kn21}
Christian Konrad and Kheeran~K. Naidu.
\newblock On two-pass streaming algorithms for maximum bipartite matching.
\newblock In Mary Wootters and Laura Sanit{\`{a}}, editors, {\em Approximation,
  Randomization, and Combinatorial Optimization. Algorithms and Techniques,
  {APPROX/RANDOM} 2021, August 16-18, 2021, University of Washington, Seattle,
  Washington, {USA} (Virtual Conference)}, volume 207 of {\em LIPIcs}, pages
  19:1--19:18. Schloss Dagstuhl - Leibniz-Zentrum f{\"{u}}r Informatik, 2021.

\bibitem{kns23}
Christian Konrad, Kheeran~K. Naidu, and Arun Steward.
\newblock Maximum matching via maximal matching queries.
\newblock In Petra Berenbrink, Patricia Bouyer, Anuj Dawar, and
  Mamadou~Moustapha Kant{\'{e}}, editors, {\em 40th International Symposium on
  Theoretical Aspects of Computer Science, {STACS} 2023, March 7-9, 2023,
  Hamburg, Germany}, volume 254 of {\em LIPIcs}, pages 41:1--41:22. Schloss
  Dagstuhl - Leibniz-Zentrum f{\"{u}}r Informatik, 2023.

\bibitem{lnw14}
Yi~Li, Huy~L. Nguyen, and David~P. Woodruff.
\newblock Turnstile streaming algorithms might as well be linear sketches.
\newblock In David~B. Shmoys, editor, {\em Symposium on Theory of Computing,
  {STOC} 2014, New York, NY, USA, May 31 - June 03, 2014}, pages 174--183.
  {ACM}, 2014.

\bibitem{zaam22}
Fuheng Zhao, Divy Agrawal, Amr~El Abbadi, and Ahmed Metwally.
\newblock Spacesaving{\(\pm\)} an optimal algorithm for frequency estimation
  and frequent items in the bounded deletion model.
\newblock {\em Proc. {VLDB} Endow.}, 15(6):1215--1227, 2022.

\bibitem{zaammr23}
Fuheng Zhao, Divyakant Agrawal, Amr~El Abbadi, Claire Mathieu, Ahmed Metwally,
  and Michel de~Rougemont.
\newblock A detailed analysis of the spacesaving{\(\pm\)} family of algorithms
  with bounded deletions.
\newblock {\em CoRR}, abs/2309.12623, 2023.

\bibitem{zmwaa21}
Fuheng Zhao, Sujaya Maiyya, Ryan Weiner, Divy Agrawal, and Amr~El Abbadi.
\newblock Kll{\(\pm\)}: Approximate quantile sketches over dynamic datasets.
\newblock {\em Proc. {VLDB} Endow.}, 14(7):1215--1227, 2021.

\end{thebibliography}

\appendix 

\section{Technical Lemma}
\begin{lemma}\label{lem:tech}
    Let $a,b,n$ be integers such that $a < n$ and $b \le n - a$. Then, the following inequality holds:
    $$\frac{ {n - a \choose b}}{{n \choose b}} \le \left(1 - \frac{b}{n} \right)^a \ .$$
\end{lemma}
\begin{proof}
We compute as follows:
    \begin{align*}
       \frac{ {n - a \choose b}}{{n \choose b}} & = \frac{(n-a)! \cdot (n-b)!}{(n-a-b)! \cdot n!} \\
       & = \frac{(n-b) \cdot (n-b-1) \cdot \ldots \cdot (n-a-b+1)}{n \cdot (n-1) \cdot \ldots \cdot (n-a+1)} \\
       & \le \left( \frac{n-b}{n} \right)^{a} = \left( 1 - \frac{b}{n} \right)^a \ ,
    \end{align*}
    where we used the inequality $\frac{n-b-j}{n-j} \le \frac{n-b}{n}$, for every $0 \le j \le a$.
\end{proof}

\end{document}